\documentclass[11pt]{article}

\usepackage[letterpaper,margin=1.00in]{geometry}

\usepackage{amsmath, amssymb, amsthm, amsfonts}
\usepackage{thmtools} 
\usepackage{thm-restate}

\usepackage{microtype}
\usepackage{cite}
\usepackage{framed}
\usepackage[framemethod=tikz]{mdframed}
\usepackage{appendix}
\usepackage{graphicx}
\usepackage{color}
\usepackage{varwidth}
\usepackage[boxed,noend,linesnumbered]{algorithm2e}
\usepackage[noend]{algpseudocode}
\usepackage[textsize=tiny]{todonotes}
\usepackage{array}
\usepackage{dsfont}

\usepackage{xspace}
\usepackage{xifthen}

\definecolor{darkgreen}{rgb}{0,0.5,0}
\definecolor{darkblue}{rgb}{0,0,0.8}
\usepackage{hyperref}
\hypersetup{
    unicode=false,          
    colorlinks=true,        
    linkcolor=darkblue,          
    citecolor=darkgreen,        
    filecolor=magenta,      
    urlcolor=cyan           
}
\RequirePackage[]{silence}
\usepackage{orcidlink}

\usepackage[capitalize, nameinlink]{cleveref} 

\newcommand{\para}[1]{\paragraph{#1}}

\newtheorem{theorem}{Theorem}[section]
\newtheorem{lemma}[theorem]{Lemma}
\newtheorem{claim}[theorem]{Claim}
\newtheorem{corollary}[theorem]{Corollary}

\newtheorem{definition}{Definition}[section]

\newcommand{\calL}{\ensuremath{\mathcal{L}}}

\newcommand{\calA}{\ensuremath{\mathcal{A}}}

\newcommand{\calC}{\ensuremath{\mathcal{C}}}

\newcommand{\stackedbinom}[3]{\left(\substack{\binom{#1}{#2}\\#3}\right)}

\newcommand{\Nout}{N^{\mathit{out}}}

\newcommand{\ignore}[1]{}

\algnewcommand\algorithmicswitch{\textbf{switch}}
\algnewcommand\algorithmiccase{\textbf{case}}

\algdef{SE}[SWITCH]{Switch}{EndSwitch}[1]{\algorithmicswitch\ #1\ \algorithmicdo}{\algorithmicend\ \algorithmicswitch}%
\algdef{SE}[CASE]{Case}{EndCase}[1]{\algorithmiccase\ #1}{\algorithmicend\ \algorithmiccase}%
\algtext*{EndSwitch}%
\algtext*{EndCase}%

\newcommand{\CONGEST}{\ensuremath{\mathsf{CONGEST}}\xspace}

\newcommand{\LOCAL}{\ensuremath{\mathsf{LOCAL}}\xspace}

\newcommand{\eps}{\varepsilon}
\renewcommand{\epsilon}{\varepsilon}
\newcommand{\poly}{\operatorname{\text{{\rm poly}}}}

\newcommand{\set}[1]{\left\{#1\right\}}

\DeclareMathOperator{\polylog}{\poly\log}

\newcommand{\hide}[1]{}

\renewcommand{\phi}{\varphi}

\begin{document}

\title{List Defective Colorings: Distributed Algorithms and Applications\footnote{This work was partially supported by the German Research Foundation (DFG) under the project number 491819048.}} 
\date{}
\author{
   Marc Fuchs \orcidlink{0000-0003-2272-4483} \\
   \small{University of Freiburg} \\
   \small{marc.fuchs@cs.uni-freiburg.de}
   \and
   Fabian Kuhn \orcidlink{0000-0002-1025-5037}\\
   \small{University of Freiburg} \\
   \small{kuhn@cs.uni-freiburg.de}
   }

\maketitle
\begin{abstract}
The distributed coloring problem is at the core of the area of distributed graph algorithms and it is a problem that has
seen  tremendous progress over the last few years. Much of the remarkable recent progress on deterministic distributed
coloring algorithms is based on two main tools: a) defective colorings in which every node of a given color can have a limited number of neighbors of the same color and b) list coloring, a natural generalization of the standard coloring problem that naturally appears when colorings are computed in different stages and one has to extend a previously computed partial coloring to a full coloring.

In this paper, we introduce \emph{list defective colorings}, which can be seen as a generalization of these two coloring variants. Essentially, in a list defective coloring instance, each node $v$ is given a list of colors $x_{v,1},\dots,x_{v,p}$ together with a list of defects $d_{v,1},\dots,d_{v,p}$ such that if $v$ is colored with color $x_{v, i}$, it is allowed to have at most $d_{v, i}$ neighbors with color $x_{v, i}$.

We highlight the important role of list defective colorings by showing that faster list defective coloring algorithms would directly lead to faster deterministic $(\Delta+1)$-coloring algorithms in the \LOCAL model. Furthermore, we extend a recent distributed list coloring algorithm by Maus and Tonoyan [DISC '20]. Slightly simplified, we show that if for each node $v$ it holds that $\sum_{i=1}^p \big(d_{v,i}+1)^2 >\deg_G^2(v)\cdot \polylog\Delta$, then this list defective coloring instance can be solved in a communication-efficient way in only $O(\log\Delta)$ communication rounds. This leads to the first deterministic $(\Delta+1)$-coloring algorithm in the standard \CONGEST model with a time complexity of
$O(\sqrt{\Delta}\cdot \polylog \Delta+\log^* n)$, matching the best time complexity in the \LOCAL model up to a $\polylog\Delta$ factor.
\end{abstract}
\clearpage
\section{Introduction and Related Work}
\label{sec:intro}

Distributed graph coloring is one of the core problems in the area of distributed graph algorithms. For this problem, one typically assumes that the graph $G=(V,E)$ to be colored represents a communication network of $n$ nodes with maximum (node)degree $\Delta$ and that the nodes (or edges) of $G$ must be colored in a distributed way by exchanging messages over the edges of $G$. The nodes typically interact with each other in synchronous rounds. If the size of messages is not restricted, this is known as the \LOCAL model and if in every round, every node can send an $O(\log n)$-bit message to every neighbor, it is known as the \CONGEST model~\cite{Peleg2000}. The problem was first studied by Linial in a paper that pioneered the whole area of distributed graph algorithms~\cite{Linial1987}. Linial in particular showed that coloring a ring network with $O(1)$ colors (and thus coloring a graph with $f(\Delta)$ colors) requires $\Omega(\log^* n)$ rounds and that in $O(\log^* n)$ rounds, one can color a graph with $O(\Delta^2)$ colors. Since then, there has been a plethora of work on distributed coloring algorithms, e.g., \cite{goldberg88,awerbuch89,Naor1991,SzegedyV93,panconesi96decomposition,KuhnW06,barenboim14distributed,BarenboimE10,Barenboim2016,barenboim16sublinear,fraigniaud16local,HarrisSS18,ChangLP18,Kuhn20,MausT20,Maus21,GhaffariKuhn21,HalldorssonKMT21,HalldorssonKNT22}. The most standard variant of the distributed coloring problem asks for a proper coloring of the nodes $V$ of $G$ with $\Delta+1$ colors. Note that this is what can be achieved with a simple sequential greedy algorithm.

Over the last approximately 15 years, we have seen remarkable progress on randomized and on deterministic distributed coloring algorithms. Much of the progress on deterministic algorithms (which are the focus of the present paper) has been achieved by studying and using two generalizations of the standard coloring problem, \emph{defective colorings} and \emph{list colorings}. In the following, we briefly discuss the history and significance of defective colorings and of list colorings in the context of deterministic distributed coloring algorithms. For lack of space, we do not go into details of the very early work on deterministic distributed coloring~\cite{Linial1987,goldberg88,SzegedyV93,KuhnW06}, the deterministic algorithms that directly result from computing network decomposition~\cite{awerbuch89,panconesi96decomposition,Rozhon2020,GGR2020,GGHIR23}, or the vast literature on randomized distributed coloring algorithms, e.g., \cite{Luby1986,SchneiderW10symmetry,Barenboim2016,HarrisSS18,ChangLP18,HalldorssonKMT21,HalldorssonKNT22}.

\para{Defective Coloring.} Given integers $d\geq 0$ and $c>0$, a $d$-defective $c$-coloring of a graph $G=(V,E)$ is an assignment of colors $\set{1,\dots,c}$ to the nodes in $V$ such that the subgraph induced by each color class has a maximum degree of at most $d$~\cite{lovasz66}. Defective colorings were introduced to distributed algorithms independently by Barenboim and Elkin~\cite{BarenboimE09} and by Kuhn~\cite{Kuhn2009} in 2009. Both papers give distributed algorithms to compute $d$-defective colorings with $O(\Delta^2/d^2)$ colors. The algorithm of \cite{Kuhn2009} extends the classic $O(\Delta^2)$-coloring algorithm by Linial to achieving this in $O(\log^* n)$ time. Both papers use defective colorings to compute proper colorings in a divide-and-conquer fashion, leading to algorithms for computing a $(\Delta+1)$-coloring in $O(\Delta+\log^* n)$ rounds.\footnote{Earlier algorithms were based on simple round-by-round color reduction schemes and required $O(\Delta^2 + \log^* n)$~\cite{Linial1987,goldberg88} and $O(\Delta\log\Delta + \log^* n)$ rounds~\cite{SzegedyV93,KuhnW06}, respectively.} This idea was pushed further by Barenboim and Elkin in \cite{BarenboimE10} and \cite{BarenboimE11}. In \cite{BarenboimE10}, they introduce the notion of \emph{arbdefective colorings}: Instead of decomposing a graph into color classes of bounded degree, the aim is to decompose a graph into color classes of bounded arboricity. More specifically, the output of an arbdefective $c$-coloring algorithm with arbdefect $d$ is a coloring of nodes with colors $\set{1,\dots,c}$ together with an orientation of the edges such that every node has at most $d$ outneighbors of the same color. For this more relaxed version of defective coloring, they show that for a given oriented graph with maximum outdegree $\beta$, one can efficiently compute a $d$-arbdefective coloring with $O(\beta/d)$-colors (in time $O(\beta^2/d^2\cdot\log n)$). Applying this recursively, for example, allows us to obtain a $\Delta^{1+o(1)}$-coloring in time $\poly\log\Delta \cdot\log n$. In \cite{BarenboimE11}, it is shown that for a family of graphs that includes line graphs, one can efficiently compute a standard $d$-defective coloring with only $O(\Delta/d)$ colors in time $O(\Delta^2/d^2 + \log^* n)$. This in particular implies that a $\Delta^{1+o(1)}$-edge coloring can be computed in time $\poly\log\Delta +O(\log^* n)$. All the algorithms of \cite{BarenboimE09,Kuhn2009,BarenboimE10,BarenboimE11} use defective coloring in the following basic way. If a computed defective coloring has $p$ colors, the space of available colors is divided into $p$ parts that can then be assigned to the $p$ color classes and handled in parallel on the respective lower degree/arboricity graphs. When doing this, one inherently has to use more than $\Delta+1$ colors because in each defective coloring step, the maximum degree (or outdegree) goes down at a factor that is somewhat smaller than the number of colors of the defective coloring. In \cite{BarenboimE09,Kuhn2009}, this is compensated by reducing the number of colors at the end of each recursion level. However, this leads to  algorithms with time complexity at least linear in $\Delta$.

\para{List Coloring.} The key to obtaining $(\Delta+1)$-coloring algorithms with a better time complexity is to explicitly consider the more general $(\mathit{degree}+1)$-list coloring problem. In this problem, every node $v$ receives a list $L_v$ of at least $\deg(v)+1$ colors as input and an algorithm has to properly color the graph in such a way that each node $v$ is colored with a color from its list $L_v$. Note that this problem can still be solved by a simple sequential greedy algorithm. Note also that the problem appears naturally when solving the standard $(\Delta+1)$-coloring problem in different phases. If a subset $S\subseteq V$ is already colored, then each node $v\in V\setminus S$ needs to be colored with a color that is not already taken by some neighbor in $S$. If $v$ has degree $\Delta$ and all already colored neighbors of $v$ have chosen different colors, the list of the remaining available colors for $v$ is exactly of length $\deg(v)+1$. The first paper that explicitly considered list coloring in the context of deterministic distributed coloring is by Barenboim~\cite{barenboim16sublinear}. In combination with the improved arbdefective coloring algorithm of \cite{BarenboimEG18}, the algorithm of the paper obtains a $(1+\eps)\Delta$-coloring in $O(\sqrt{\Delta}+\log^* n)$ rounds by first computing a $O(\sqrt{\Delta})$-arbdefective $O(\sqrt{\Delta})$-coloring and by subsequently iterating over the $O(\sqrt{\Delta})$ color classes of this arbdefective coloring and solving the corresponding list coloring problem in $O(1)$ time. With the same technique, the paper also gets an $O(\Delta^{3/4}+\log^* n)$-round algorithm for $(\Delta+1)$-coloring. This algorithm also works in the \CONGEST model, i.e., by exchanging messages of at most $O(\log n)$ bits. For algorithms with a round complexity of the form $f(\Delta)+O(\log^* n)$, this still is the fastest known $(\Delta+1)$-coloring algorithm in the \CONGEST model. The algorithm was improved by Fraigniaud, Kosowski, and Heinrich~\cite{fraigniaud16local}. In combination with the subsequent results of \cite{BarenboimEG18,MausT20}, the algorithm of \cite{fraigniaud16local} leads to an $O(\sqrt{\Delta\log\Delta}+\log^* n)$-round distributed algorithm for $(\mathit{degree}+1)$-list coloring and thus also for $(\Delta+1)$-coloring. As one of the main results of this paper, we give a \CONGEST algorithm that almost matches this and that has a time complexity of $O(\sqrt{\Delta}\poly\log\Delta + \log^*n)$. List colorings and defective colorings have also been explicitly used in all later deterministic distributed coloring algorithms~\cite{Kuhn20,BalliuKO20,GhaffariKuhn21,BalliuBKO22}. We next discuss an idea that was introduced in \cite{Kuhn20} and that is particularly important in the context of the present paper.

\para{Distributed Color Space Reduction.}
The objective of \cite{Kuhn20} was to extend the coloring algorithms of \cite{BarenboimE10,BarenboimE11} to list colorings. The algorithms of \cite{BarenboimE10,BarenboimE11} are based on computing arbdefective or defective colorings to recursively divide the graph into low (out)degree parts that use disjoint sets of colors. This leads to fast coloring algorithms, however, the number of required colors grows exponentially with the number of recursion levels. While it is not clear how to efficiently turn a standard distributed coloring algorithm that uses significantly more than $\Delta+1$ colors into a $(\Delta+1)$-coloring algorithm, by using the techniques introduced in \cite{barenboim16sublinear,fraigniaud16local}, we can do this if we have a list coloring algorithm. Essentially, if we have a list coloring algorithm that uses lists of size $O(\alpha(\Delta+1))$, it can be turned into a $(\mathit{degree}+1)$-list coloring algorithm in only $\tilde{O}(\alpha^2)$ rounds (and in some cases even in $O(\alpha)$ rounds). However, if the nodes have different lists, a defective coloring does not easily split the graph into independent coloring problems. As a generalization of defective colorings, \cite{Kuhn20} introduces a tool called \emph{color space reduction}. Assuming that all lists consist of colors of some color space $\calC$. For a given partition of $\calC$ into disjoint parts $\calC_1,\dots,\calC_p$, a color space reduction algorithm partitions the nodes $V$ into $p$ parts $V_1,\dots,V_p$ such that for every node $v$, with $v\in V_i$ and $v$ has $\deg_i(v)$ neighbors in $V_i$, the algorithm tries to keep the ratio $|L_v\cap \calC_i|/\deg_i(v)$ as close as possible to the initial list-degree ratio $|L_v|/\deg(v)$. In \cite{Kuhn20}, it is shown that the arbdefective and defective coloring algorithms of \cite{BarenboimE10,BarenboimE11} can be generalized to compute a color space reduction. If the size of the color space is polynomial in $\Delta$, this leads to $(\mathit{degree}+1)$-coloring algorithms with time complexities of $2^{O(\sqrt{\log \Delta})}\log n$ in general graphs and of $2^{O(\sqrt{\log n})}+O(\log^* n)$ in graphs of bounded neighborhood independence, a family of graphs that includes line graphs of bounded rank hypergraphs. The complexity of the $(\mathit{degree}+1)$-edge coloring problem was later improved to $(\log\Delta)^{O(\log\log\Delta)}+O(\log^* n)$ in \cite{BalliuKO20} and to $\poly\log\Delta + O(\log^* n)$ in \cite{BalliuBKO22}. In both cases, this was achieved by designing better distributed color space reduction algorithms for line graphs. The $(\Delta+1)$-coloring algorithm of \cite{Kuhn20} for general graphs was later subsumed by a deterministic $O(\log^2\Delta\cdot\log n)$-round algorithm for the $(\Delta+1)$-coloring problem in \cite{GhaffariKuhn21}. 

\para{List Defective Colorings.}
Color space reductions can be seen as a special case of the following list variant of defective colorings. Each node $v$ has a list of possible colors that it can choose (e.g., which color subspace $\calC_i$ to use). Depending on what color $v$ chooses, it can tolerate different defects (e.g., depending on the size $|L_{v,i}\cap \calC_i|$ of the remaining color list when choosing color subspace $\calC_i$). In the following, we formally define \emph{list defective colorings}. One of the objectives of this paper is to understand the relation of list defective colorings to each other and to other coloring problems, and we will see in particular that better list defective coloring algorithms can directly lead to better algorithms for standard coloring problems.

In a list defective coloring problem, as input, each node $v$  obtains a \emph{color list} $L_v\subseteq \calC$, where $\calC$ is the space of possible colors. Each node $v$ further has a defect function $d_v:L_v \to \mathbb{N}_0$ that assigns a non-negative integral defect value to each color in $v$'s list $L_v$. Given vertex lists $L_v$, a list vertex coloring is an assignment $\varphi:V\to \calC$ that assigns each node $v\in V$ a color $\varphi(v)\in L_v$. In the following, we formally define three variants of list defective coloring.

\begin{definition}[List Defective Coloring]\label{def:LDC}
  Let $G=(V,E)$ be a graph, let $\calC$ be a color space, and assume that each node $v\in V$ is given a color list $L_v \subseteq \calC$ and a defect function $d_v:L_v\to \mathbb{N}_0$. Further, assume that we are given a list vertex coloring $\varphi:V\to \calC$.
  \begin{itemize}
  \item The coloring $\varphi$ is a \emph{list defective coloring} iff every $v\in V$ has at most $d_v(\varphi(v))$ neighbors of color $\varphi(v) \in L_v$.
  \item If $G$ is a directed graph, $\varphi$ is called an \emph{oriented list defective coloring} iff every $v\in V$ has at most $d_v(\varphi(v))$ out-neighbors of color $\varphi(v) \in L_v$.
  \item In combination with an edge orientation $\sigma$, $\varphi$ is called a \emph{list arbdefective coloring} iff it is an oriented list defective coloring w.r.t.\ the directed graph induced by the edge orientation $\sigma$.
  \end{itemize}
  An (oriented) list (arb)defective coloring is called an \emph{(oriented) $p$-list (arb)defective coloring} for some integer $p>0$ if for all $v\in V$, $|L_v|\leq p$.
\end{definition}

Note that the difference between an oriented list defective coloring and a list arbdefective coloring is that in an oriented list defective coloring, the edge orientation of $G$ is given as part of the input and in a list arbdefective coloring, the edge orientation is a part of the output. By a result of Lov\'{a}sz~\cite{lovasz66}, it is well-known that a $d$-defective $c$-coloring of a graph $G$ with maximum degree $\Delta$ always exists if $c(d+1)>\Delta$ (see \cref{lemma:defectiveexistence} for details). Note that this condition is also necessary if $G=K_{\Delta+1}$. By computing a balanced orientation of the edges of each color class of such a coloring, one can also deduce that a $d$-arbdefective $c$-coloring always exists if $c(2d+1)>\Delta$ (see \cref{lemma:arbdefectiveexistence} for details). Again, this condition is necessary if $G=K_{\Delta+1}$. By generalizing the potential function argument of \cite{lovasz66}, in \Cref{sec:existential}, we prove that the natural generalization of both existential statements also holds for the respective list defective coloring variants. Specifically, we show that for given color lists $L_v$ and defect functions $d_v$, a list defective coloring always exists if
\begin{equation}\label{eq:existential1}
  \forall v\in V\,:\,\sum_{x\in L_v} \big(d_v(x)+1\big) > \Delta
\end{equation}
and a list arbdefective coloring always exists if 
\begin{equation}\label{eq:existential2}
  \forall v\in V\,:\,\sum_{x\in L_v} \big(2d_v(x)+1\big) > \Delta.
\end{equation}
Both conditions are necessary if the graph is a $(\Delta+1)$-clique and if all nodes have the same color list and the same defect function. For arbdefective colorings, it has further been shown in \cite{balliu2021hideandseek} that Condition \eqref{eq:existential1} is necessary and sufficient to compute such colorings in time $f(\Delta) + O(\log^* n)$. Whenever \eqref{eq:existential1} does not hold, there is an $\Omega(\log_\Delta n)$-round lower bound for deterministic distributed list arbdefective coloring algorithms.

 \subsection{Our Contributions}
 \label{sec:contributions}

In the following, whenever we consider a graph $G=(V,E)$, we assume that $n$ denotes the number of nodes of $G$, $\Delta$ denotes the maximum degree of $G$, and $\deg(v)$ denotes the degree of a node $v$. Also, if $G$ is a directed graph, $\beta_v$ refers to the outdegree of node $v$ and $\beta$ denotes the maximum outdegree. Further, if we discuss any list defective coloring problem, unless stated otherwise, we assume that the colors come from space $\calC$, $L_v\subseteq \calC$ denotes the list of node $v$, and $d_v:L_v\to \mathbb{N}_0$ denotes the defect function of node $v$. We further assume that $\Lambda:=\max_{v\in V} |L_v|$ denotes the maximum list size. As it is common in the distributed setting, we do not analyze the complexity of internal computations at nodes. We briefly discuss the complexity of internal computations for our algorithms in \cref{sec:InternalComplexity}. 

\para{Oriented List Defective Coloring.}
As our main technical contribution, we give an efficient deterministic distributed algorithm for computing oriented list defective colorings. This algorithm is an adaptation of the techniques developed by Maus and Tonoyan~\cite{MausT20} to obtain a $2$-round algorithm for proper vertex colorings in directed graphs of small maximum outdegree.

\begin{theorem}\label{thm:def_local_list_coloring}
  Let $G=(V,E)$ be a properly $m$-colored directed graph and assume that we are given an oriented list defective coloring instance on $G$. Assume that for every node $v\in V$, for a sufficiently large constant $\alpha>0$, it holds that

  \begin{equation}
    \label{eq:defalg_condition}
    \sum_{x\in L_v} \big(d_v(x)+1\big)^2 \geq \alpha\cdot \beta_v^2 \cdot\kappa(\beta, \calC, m),
  \end{equation}
  where $\kappa(\beta, \calC, m) =  (\log\beta + \log\log|\calC| + \log\log m) \cdot (\log\log\beta + \log\log m) \cdot \log^2\log\beta$.
  
  Then, there is a deterministic distributed algorithm that solves this oriented list defective coloring instance in $O(\log\beta)$ rounds using $O\big(\min\set{|\calC|, \Lambda\cdot \log |\calC|} + \log\beta + \log m\big)$-bit messages.
\end{theorem} 

\para{Recursive Color Space Reduction.}
We have already discussed that we can use list defective colorings to recursively divide the color space. We next elaborate on the power of doing recursive color space reduction directly for list defective coloring problems. The following theorem shows that in this way, at the cost of solving a somewhat weaker problem, we can sometimes significantly improve the time complexity or the required message size. In the following, we assume that we have an oriented list defective coloring algorithm $\calA$, where the complexity is in particular a function of the maximum list size $\Lambda$. More specifically, the following theorem specifies the properties of $\calA$ by an arbitrary parameter $\nu \geq0$ and by arbitrary non-decreasing functions $\kappa(\Lambda)$, $T(\Lambda)$, and $M(\Lambda)$. Note that the functions $\kappa(\Lambda)$, $T(\Lambda)$, and $M(\Lambda)$ can in principle also depend on other global properties such as the maximum degree $\Delta$, the maximum outdegree $\beta$, or the number of nodes $n$. When applying $\calA$ to obtain algorithm $\calA'$, we then however treat those other parameters as fixed quantities.

\begin{restatable}{theorem}{restatesecond}\label{thm:spacereduction}
  Let $\nu\geq0$ be a parameter and let $\kappa(\Lambda)$, $T(\Lambda)$, and $M(\Lambda)$ be non-decreasing functions of the maximum list size $\Lambda$. Assume that we are given a deterministic distributed algorithm $\calA$ that solves oriented list defective coloring instances for which
  \[
    \forall v\in V\,:\,\sum_{x\in L_v}\big(d_v(x)+1\big)^{1+\nu} \geq \beta_v^{1+\nu}\cdot \kappa(\Lambda).
  \]
  Assume further that the round complexity of $\calA$ is $T(\Lambda)$ and that $\calA$ requires messages of $M(\Lambda)$ bits.

  Then, for any integer $p\in(1,|\calC|]$, there exists a deterministic distributed algorithm $\calA'$ that solves oriented list defective coloring instances for which
  \[
    \forall v\in V\,:\,\sum_{x\in L_v}\big(d_v(x)+1\big)^{1+\nu} \geq \beta_v^{1+\nu}\cdot \kappa(p)^{\lceil \log_p |\calC|\rceil}
  \]
  in time $O(T(p)\cdot \log_p |\calC|)$ and with $M(p)$-bit messages.
\end{restatable} 

When replacing $\beta_v$ by $\deg(v)$, the same theorem also holds for list defective colorings (in undirected graphs). One can easily see this as a list defective coloring can be turned into an equivalent oriented list defective problem, by replacing every edge $\set{u,v}$ of an undirected graph by the two directed edges $(u,v)$ and $(v,u)$.  

Note that the number of colors of a standard defective coloring corresponds to the maximum list size $\Lambda$ of a list defective coloring, i.e., a standard defective coloring with $c$ colors is a special case of list defective coloring with lists of size $c$. Many of the existing defective and arbdefective coloring algorithms have a round complexity that is of the form $\poly(c) + O(\log^* n)$~\cite{BarenboimE11,barenboim14distributed,BarenboimEG18}. For a concrete application of \cref{thm:spacereduction}, we therefore assume  that the time complexity of algorithm $\calA$ is of the form $T(\Lambda) = \poly(\Lambda) + O(\log^* n)$. For simplicity, we further assume that the size of the color space $\calC$ is at most polynomial in $\beta$. By setting $p=2^{O(\sqrt{\log\beta \cdot \log\kappa(\Lambda)})}$, we then get an algorithm that solves oriented list defective coloring problems with $\forall v\in V\,:\,\sum_{x\in L_v}(d_v(x)+1)^{1+\nu} \geq \beta_v^{1+\nu}\cdot 2^{O(\sqrt{\log\beta \cdot \log\kappa(\Lambda)})}$ in time $2^{O(\sqrt{\log\beta \cdot \log\kappa(\Lambda)})} + O(\log^* n)$ rounds. For details, we refer to \Cref{cor:spacereduction_time} in \Cref{sec:spacereduction}.

As a second application of \Cref{thm:spacereduction}, consider the oriented list defective coloring algorithm given by \Cref{thm:def_local_list_coloring}. The round complexity of this algorithm is $O(\log\beta)$ and we cannot hope to get a time improvement by recursively subdividing the color space. We can however improve the necessary message size. The message size of the algorithm is essentially linear in the maximum list size $\Lambda$. If we choose $p\ll\Lambda$, the message size becomes essentially linear in $p$. Assume for example that $\Lambda$ and the color space are both polynomial in $\beta$. We then only need a constant number of recursion levels to reduce the message size to $O(\beta^{\eps} + \log m)$ for any constant $\eps>0$ (see \Cref{cor:spacereduction_msg}). We will apply this idea to obtain our new \CONGEST algorithm for the $(\Delta+1)$-coloring problem.

\para{Degree + 1 and List Arbdefective Colorings.}
The remaining contributions deal with applying (oriented) list defective coloring algorithms to solve the standard $(\mathit{degree}+1)$-coloring problem and more general other coloring problems. The following theorem shows that in combination with (oriented) list defective coloring algorithms, the general technique of \cite{barenboim16sublinear,fraigniaud16local} cannot only be used to solve standard $(\mathit{degree}+1)$-coloring instances, but more generally also to solve list arbdefective coloring instances for which for all nodes $v$, $\sum_{x\in L_v} (d_v(x)+1) \geq \deg(v) + 1$. Further, if we assume (oriented) list defective coloring algorithms of a certain quality (which is better than what we currently know), we directly obtain algorithms that potentially significantly improve the state of the art for the standard $(\mathit{degree}+1)$-coloring problem. For the following theorem, we assume that for two parameters $\nu>0$ and $\kappa>0$, we have an oriented list defective coloring algorithm $\calA_{\nu,\kappa}^O$ or a list defective coloring algorithm $\calA_{\nu,\kappa}^D$ to solve instances for which for all $v$, $\sum_{x\in L_v} (d_v(x)+1)^{1+\nu} \geq \beta_v^{1+\nu}\cdot \kappa$ or $\sum_{x\in L_v} (d_v(x)+1)^{1+\nu} \geq \deg(v)^{1+\nu}\cdot \kappa$. We use $T_{\nu,\kappa}^O$ and $T_{\nu,\kappa}^D$ to denote the time complexities of the two algorithms. Note that the parameter $\kappa$ can depend (monotonically) on global properties such as the maximum list size $\Lambda$, the maximum degree $\Delta$, or the maximum outdegree $\beta$. We then however treat those global parameters as fixed quantities when applying the algorithms $\calA_{\nu,\kappa}^O$ and $\calA_{\nu,\kappa}^D$ recursively.

\begin{restatable}{theorem}{restatethird}\label{thm:arbdefective}
  Let $\nu\geq 0$ and $\kappa>0$ be two parameters, let $G=(V,E)$ be an undirected graph with maximum degree $\Delta$, and assume that we are given a list arbdefective coloring instance of $G$ for which $\forall v\in V\,:\, \sum_{x\in L_v}(d_v(x)+1) > \deg(v)$. Using the oriented list defective coloring algorithm $\calA_{\nu,\kappa}^O$, the given list arbdefective coloring problem can be solved in $O\big(\Lambda^{\frac{\nu}{1+\nu}}\cdot \kappa^{\frac{1}{1+\nu}}\cdot \log(\Delta)\cdot T_{\nu,\kappa}^O + \log^* n\big)$ rounds. Using the list defective coloring algorithm $\calA_{\nu,\kappa}^D$, the given list arbdefective coloring problem can be solved in $O\big(\Lambda^{\nu}\cdot \kappa^2\cdot \log(\Delta)\cdot T_{\nu,\kappa}^D + \log^* n\big)$ rounds. If $\nu\geq \nu_0$ for some constant $\nu_0>0$, in both time bounds, the $\log(\Delta)$ term can be substituted by $\log(\Delta/\Lambda)$. If $\calA_{\nu,\kappa}^D$ (or $\calA_{\nu,\kappa}^O$) uses messages of at most $B$ bits, then the resulting list arbdefective coloring algorithm uses messages of $O(B + \log n)$ bits.
\end{restatable}

Note that the algorithm of \Cref{thm:def_local_list_coloring} satisfies the requirements of algorithm $\calA_{\nu,\kappa}^O$ for $\nu=1$. If we assume that we first compute an $O(\Delta^2)$-coloring of $G$ in time $O(\log^* n)$ by using a standard algorithm of \cite{Linial1987} and if we assume the size of the color space is at most exponential in $\Delta$, then $\kappa=O(\log\Delta\cdot\log^3\log\Delta)$ and $T_{\nu,\kappa}^O=O(\log\Delta)$. 
When using the algorithm of \Cref{thm:def_local_list_coloring} as algorithm $\calA_{\nu,\kappa}^O$ in \Cref{thm:arbdefective}, \Cref{thm:arbdefective} therefore implies that the given arbdefective coloring instance can be solved in $O\big(\sqrt{\Lambda}\cdot \log^{3/2}\Delta \cdot \log^{3/2}\log\Delta + \log^* n\big)$ rounds. Hence, the theorem in particular entails that a $d$-arbdefective $\lfloor \frac{\Delta}{d+1} + 1\rfloor$-coloring can be computed in $O(\sqrt{\Delta/(d+1)}\cdot \log^{3/2}\Delta\cdot\log^{3/2}\log\Delta + \log^* n)$ rounds. Even when using $O(\Delta/d)$ colors, the best previous algorithm for this problem required $O(\Delta/d + \log^* n)$ rounds~\cite{BarenboimEG18}. \Cref{thm:arbdefective} further shows that if we could get a fast oriented list defective coloring algorithm for a condition of the form $\sum_{x\in L_v}(d_v(x) +1)^{2-\eps}\geq \beta_v^{2-\eps}\poly\log\Delta$, we would already significantly improve the existing $O(\sqrt{\Delta\log\Delta} + \log^* n)$-round algorithm of \cite{fraigniaud16local,BarenboimEG18,MausT20} to compute a $(\Delta+1)$-coloring. The same would be true in case we could get a fast list defective coloring algorithm for a condition of the form $\sum_{x\in L_v}(d_v(x)+1)^{3/2-\eps} \geq \deg(v)^{3/2-\eps}\cdot\poly\log\Delta$. This indicates that the current obstacle for significantly improving the $O(\sqrt{\Delta\log\Delta} + \log^* n)$-round algorithm (in case this is possible) is to improve our understanding of the complexity of computing defective colorings and possible list defective colorings.

\para{Faster Coloring in the \boldmath\CONGEST Model.} Finally, we show that by combining \Cref{thm:def_local_list_coloring,thm:spacereduction,thm:arbdefective}, we obtain a faster $(\mathit{degree}+1)$-list coloring for the \CONGEST model. \newpage

\begin{restatable}{theorem}{restatecongest}\label{thm:CONGEST}
  Let $G=(V,E)$ be an $n$-node graph and assume that we are given a $(\mathit{degree}+1)$-list coloring instance on $G$. If the color space $\mathcal{C}$ of the problem is of size $|\mathcal{C}|\leq \poly(\Delta)$, there exists a deterministic \CONGEST algorithm for solving the $(\mathit{degree}+1)$-list coloring instance in time $\sqrt{\Delta} \cdot \polylog \Delta + O(\log^* n)$. If the color space is of size $|\mathcal{C}| = O(\Delta)$ (such as, e.g., for the standard $(\Delta+1)$-coloring problem), the time complexity of the algorithm is $O(\sqrt{\Delta}\cdot\log^2\Delta\cdot\log^6\log\Delta + \log^* n)$.
\end{restatable}

Note that there is a $O(\log^2\Delta\cdot\log n)$-round \CONGEST algorithm for solving $(\mathit{degree}+1)$-list coloring instances~\cite{GhaffariKuhn21}. Further, by \cite{fraigniaud16local,BarenboimEG18,MausT20} there is an $O(\sqrt{\Delta \log \Delta} + \log^* n)$ algorithm for the problem that uses messages of size $\tilde{O}(\Delta)$. Thus, $(\Delta+1)$-coloring algorithms running in $\sqrt{\Delta} \cdot \polylog \Delta + O(\log^* n)$ rounds in the \CONGEST model are already known as long as $\Delta = O(\log n)$ or $\Delta = \Omega(\log^2 n)$. Our results fill this gap and give such an algorithm for $\Delta \in [\omega(\log n), o(\log^2 n)]$. A rough explanation of why the previous deterministic \CONGEST algorithms fail to compute $(\Delta+1)$-colorings efficiently when $\Delta \in [\omega(\log n), o(\log^2 n)]$ is the following. In the algorithm of \cite{fraigniaud16local,BarenboimEG18,MausT20}, every node has to \textit{learn} the color lists of its neighbors, which requires that $\Omega(\Delta \cdot \log \Delta)$ bits have to be sent over every edge (which only works in \CONGEST if $\Delta = O(\log n)$). For other algorithms (such as for the algorithm of \cite{GhaffariKuhn21}), the round complexity of the algorithms is at least $\Omega(\log n)$, which only leads to efficient time complexities in $\tilde{O}(\sqrt{\Delta})$ if $\Delta = \Omega(\log^2 n)$. 

\subsection{Organization of the paper} The remainder of the paper is organized as follows. In \Cref{sec:model}, we formally define the communication model and we introduce the necessary mathematical notations and definitions. \Cref{sec:OLDC} is the main technical section. It discusses our oriented list defective coloring algorithms, leading to the proof of \Cref{thm:def_local_list_coloring}. Subsequently, \Cref{sec:spacereduction} discusses how to improve existing list (defective) coloring algorithms by recursively reducing the color space. \Cref{sec:arbdefective} then shows how (oriented) list defective coloring algorithms can be applied to efficiently solve the standard $(\mathit{degree}+1)$-coloring problem and even list arbdefective coloring problems. The section also shows how this, in combination with the results in \Cref{sec:OLDC} and \Cref{sec:spacereduction}, leads to our new $(\mathit{degree}+1)$-coloring algorithm for the \CONGEST model. 


\section{Model and Preliminaries}
\label{sec:model}

\para{Communication Model.} In the \LOCAL model and the \CONGEST model~\cite{Peleg2000}, the network is abstracted as an $n$-node graph $G=(V, E)$ in which each node is equipped with a unique $O(\log n)$-bit identifier. Communication between the nodes of $G$ happens in synchronous rounds. In every round, every node of $G$ can send a potentially different message to each of its neighbors, receive the messages from the neighbors and perform some arbitrary internal computation. Even if $G$ is a directed graph, we assume that communication can happen in both directions. All nodes start an algorithm at time $0$ and the time or round complexity of an algorithm is defined as the total number of rounds needed until all nodes terminate (i.e., output their color in a coloring problem). In the \LOCAL model, nodes are allowed to exchange arbitrary messages, whereas in the \CONGEST model, messages must consist of at most $O(\log n)$ bits.

\para{Mathematical Notation.} Let $G=(V,E)$ be a graph. Throughout the paper, we use $\deg_G(v)$ to denote the degree of a node $v\in V$ in $G$ and $\Delta(G)$ to denote the maximum degree of $G$. If $G$ is a directed graph, we further use $\beta_{v,G}$ to denote the outdegree of a node $v\in V$. More specifically, for convenience, we define $\beta_{v,G}$ as the maximum of $1$ and the outdegree of $v$, i.e., we also set $\beta_{v,G}=1$ if the outdegree of $v$ is $0$. The maximum outdegree $\beta_G$ of $G$ is defined as $\beta_G:=\max_{v\in V}\beta_{v,G}$. We further use $N_G(v)$ to denote the set of neighbors of a node $v$ and if $G$ is a directed graph, we use $\Nout_G(v)$ to denote the set of outneighbors of $v$. In all cases, if $G$ is clear from the context, we omit the subscript $G$. When discussing one of the list defective coloring problems on a graph $G=(V,E)$, we will typically assume that $\calC$ denotes the space of possible colors, and we use $L_v$ and $d_v$ for $v\in V$ to denote the color list and defect function of node $v$. Throughout the paper, we will w.l.o.g.\ assume that $\calC\subseteq \mathbb{N}$ is a subset of the natural numbers. When clear from the context, we do not explicitly introduce this notation each time. Further, for convenience, for an integer $k\geq 1$, we use $[k] := \set{1,\dots ,k}$ to denote the set of integers from $1$ to $k$. Further, for a finite set $A$ and an integer $k\geq 0$, we use $2^A$ to denote the power set of $A$ and $\binom{A}{k}$ to denote the set of subsets of size $k$ of $A$. Finally, we use $\log(x):=\log_2(x)$ and $\ln(x) :=\log_e(x)$.


\section{Distributed Oriented List Defective Coloring Algorithms}
\label{sec:OLDC}

\subsection{Fundamentals}
\label{sec:overview}

Our algorithm is based on the list coloring approach of Maus and Tonoyan~\cite{MausT20} that we sketch next. As input, each node of $G=(V,E)$ obtains a color list $L_v \subseteq \calC$ of size $|L_v|\geq \alpha\beta^2 \tau$ for a sufficiently large constant $\alpha>0$ and some integer parameter $\tau>0$. In addition, the nodes are equipped with an initial proper $m$-coloring of $G$. The ``highlevel'' idea is based on the classic one-round $O(\beta^2\log m)$-coloring algorithm of Linial~\cite{Linial1987}. As an intermediate step of the algorithm of \cite{MausT20}, every node $v$ chooses a subset $C_v\subseteq L_v$ of size $|C_v|=\beta\tau$ of its list such that for every outneighbor $u$ of $v$, it holds that $|C_u\cap C_v|<\tau$. To obtain a proper coloring of $G$, node $v$ can then choose a color $x\in C_v$ that does not appear in any of the sets $C_u$ of one of the $\leq\beta$ outneighbors $u$ of $v$. If all nodes have to pick a color from $\set{1,\dots,\alpha\beta^2\tau}$ and if $\tau=O(\log \beta + \log m)$ is chosen sufficiently large, Linial shows that such sets $C_v$ for all nodes $v$ can be computed in $0$ rounds without communication. However, this is not true for the list coloring variant of the problem considered in \cite{MausT20}.

Before we show how the authors of \cite{MausT20} overcome the problems of list, we introduce some terminology. Let $P_0$ be the original list coloring problem that we need to solve and let $P_1$ be the intermediate problem of choosing a set $C_v$ from $\binom{L_v}{\beta \cdot \tau}$ s.t. $|C_v \cap C_u| < \tau$ for all outneighbors $u$ of $v$. As discussed, after solving $P_1$, $P_0$ can be solved in a single round, each node $v$ just needs to learn the sets $C_u$ of all its outneighbors $u$. To solve $P_1$, the authors of \cite{MausT20} introduce a new problem $P_2$, that can be seen as a ``higher-dimensional'' variant of Linial's algorithm. $P_2$ is defined in such a way that it can be solved without communication in $0$ rounds and such that after solving $P_2$, $P_1$ can be solved in a single round. In problem $P_2$, every node $v$ computes a set of possible candidates for the set $C_v$. For a more detailed description we need to define the following conflict relation.

\begin{definition}[Conflict relation $\Psi(\tau',\tau)$]
  \label{def:PsiConflict}
  Let $\tau',\tau>0$ be two parameters. The relation $\Psi(\tau',\tau)\subseteq 2^{2^{\calC}}\times 2^{2^{\calC}}$ is defined as follows. For any $K_1,K_2\in 2^{2^{\calC}}$, we have

  \begin{eqnarray*}
    (K_1,K_2) \in \Psi(\tau',\tau)  \Leftrightarrow \exists\text{ distinct } C_1,\dots,C_{\tau'}\in K_1\ \text{s.t.}\ \\
    \forall i\in \set{1,\dots,\tau'}\, \exists\,  C\in K_2\text{ for which } |C_i\cap C|\geq \tau.
  \end{eqnarray*}
\end{definition}

In a solution to problem $P_2$, every node $v$ outputs a set $K_v\subseteq 2^{\binom{L_v}{\beta\tau}}$ such that $|K_v|=\beta \tau'$ (for some integer parameter $\tau'>0$) and such that for every outneighbor $u$ of $v$, $(K_v,K_u)\not\in\Psi(\tau',\tau)$. Note that this implies that for every outneighbor $u$, $K_v$ contains at most $\tau'-1$ sets $C$ for which there is a set $C'\in K_u$ for which $|C\cap C'|\geq \tau$. Because $K_v$ has size $\beta\tau'$, there exists some $C_v\in K_v$ such that for every $C'\in K_u$ for every outneighbor $u$, we have $|C_v \cap C'|<\tau$. A solution of $P_2$ can be transformed into a solution of $P_1$ in a single round (each node $v$ has to communicate its set $K_v$ to all its outneighbors $u$). Maus and Tonoyan~\cite{MausT20} showed that for appropriate choices of the parameters $\tau$ and $\tau'$, $P_2$ can be solved in $0$ rounds. To see this, consider the following technical lemma, which follows almost directly from two lemmas in \cite{MausT20}.

\begin{lemma}[adapted from \cite{MausT20}]
  \label{lemma:InterfaceToYannic}
  Let $\gamma,\tau,\tau'\geq 1$ be three integer parameters such that $\tau\geq 8\log\gamma + 2\log\log|\calC| + 2\log\log m + 16$ and $\tau' = 2^{\tau - \lceil\log(2e\gamma^2)\rceil}$. For every color list $L\in \binom{\calC}{\ell}$ of size $\ell$ for some $\ell\geq 2e\gamma^2\tau$, we further define $S(L) := \stackedbinom{L}{\gamma\tau}{\gamma\tau'}$.
  Then, there exists $\bar{S}(L)\subseteq S(L)$ such that $|\bar{S}(L)|\geq |S(L)|/2$ and such that for every $K\in \bar{S}(L)$ and every $L'\in \binom{\calC}{\ell}$, there are at most $d_2<\frac{1}{4m|\calC|^{\ell}}\cdot|S(L)|$ different $K'\in S(L')$ such that $(K, K') \in \Psi(\tau',\tau)$ or $(K', K) \in \Psi(\tau',\tau)$. Further, for every $K\in S(L)$ and every $L'\in \binom{\calC}{\ell}$, there are at most $d_2$ different $K'\in S(L')$ for which $(K,K')\in\Psi(\tau',\tau)$. 
\end{lemma}
\begin{proof}
  The proof is based on Lemma 3.3 and 3.4 in \cite{MausT20}. The details of the proof are given in \cref{sec:adaptionOf}.
\end{proof}

Let us sketch how \Cref{lemma:InterfaceToYannic} implies that for appropriate choices of the parameters, $P_2$ can be solved  without communication. In the following, we set the parameters of \Cref{lemma:InterfaceToYannic} as $\gamma:=\beta$ and $\tau$ and $\tau'$ as given by the lemma statement. Assume that initially, every node $v$ has a list $L_v$ of size $\ell$, where $\ell\geq 2e \beta^2 \tau$. We define the type $T_v$ of a node as the tuple $(c,L_v)$, where $c$ is the color of $v$ in the initial proper $m$-coloring of $G$ and $L_v\in \binom{\calC}{\ell}$ is the color list of $v$. Let $T_1, \dots, T_t$ be a fixed ordering of the $t = m \binom{|\calC|}{\ell} \leq m |\calC|^\ell$ types and let $L_i$ be the color list of nodes of type $T_i$.  If we assign a set $K_i\in S(L)$ to each type $T_i$ so that for any two sets $K_i$ and $K_j$, $(K_i,K_j)\not\in\Psi(\tau', \tau)$ and $(K_j, K_i)\not\in\Psi(\tau', \tau)$, then if each node $v$ of type $T_i$ (for $i\in \set{1,\dots,t}$) outputs $K_i$, this assignment solves problem $P_2$. We assign the sets $K_i$ greedily, where for every type $T_i$, we choose some $K_i\in \bar{S}(L_i)$ such that $\bar{S}(L_i)$ is the subset of $S(L_i)$ that is guaranteed to exist by \Cref{lemma:InterfaceToYannic}. Assume for any $i \geq 1$ each type $T_j$ for $j\in \set{1,\dots,i-1}$ already picked $K_j$. Then type $T_i$ will pick some list $K_i \in \bar{S}(L)$ that does not conflict with choices of the $i-1$ previous types. By the lemma, for any type $T_j$, $j\neq i$, there are at most $d_2<\frac{1}{4m|\calC|^{\ell}}\cdot|S(L_i)|\leq \frac{1}{2t}\cdot|\bar{S}(L_i)|$ sets in $\bar{S}(L_i)$ that conflict. Because there are only $t$ types, we can always choose an appropriate $K_i$ that does not conflict with already assigned sets $K_j$ for $j<i$. Consequently, $P_2$ can be solved in $0$ rounds, and thus the original list coloring problem can be solved in $2$ rounds. A more formal and also more general proof of this zero-round solvability of $P_2$ is given in Lemma 2.1 of \cite{MausT20}.


\subsection{Basic Oriented List Defective Coloring Algorithm}
\label{sec:basicAlgo}

In the following, we first give a basic algorithm that solves a slightly generalized version of the OLDC problem. Concretely, the algorithm assigns a color $x_v\in L_v$ with defect $d_v(x_v)$ to each node $v$ such that at most $d_v(x_v)$ outneighbors $w$ of $v$ choose a color $x_w$ with $|x_v - x_w| \leq g$, where $g\geq 0$ is some given parameter. Recall that we assume that all colors are integers and therefore, the value $|x_v-x_w|$ is defined. Note that for $g=0$, this is the OLDC problem as defined in \Cref{def:LDC}.  We give a basic algorithm for this more general problem because we will need it as a subroutine in the algorithm for proving our main technical result, \Cref{thm:def_local_list_coloring}. The steps for solving the generalized OLDC problem are similar to the approach described in \Cref{sec:overview}. We, however, in particular have to adapt the algorithm to handle the case where each node comes with an individual list size.

\para{A single defect per node.} At the core of our basic (generalized) OLDC algorithm is an algorithm that solves the following weaker variant of the problem. Instead of having color-specific defects, every node $v$ has a fixed defect value $d_v\geq 0$, i.e., we have $d_v(x)=d_v$ for all $x\in L_v$. Based on an algorithm for this \textit{single-defect} case, one can solve the general OLDC problem by using a reduction explained in the proof of \cref{lemma:SimpleOLDC}. For that reason, assume during the following section that each node $v$ has three given inputs, a color list $L_v$, a defect value $d_v \geq 0$ and the number of outneighbors $\beta_v$. For each node $v$, the algorithm then requires lists of size $|L_v| \geq \alpha(\beta_v/(d_v+1))^2\cdot \tau$ for some constant $\alpha>0$ and some parameter $\tau>0$. Note that the required list size only depends on the ratio between $\beta_v$ and $d_v+1$ and not on their actual values. In the following section, we therefore do not work with the defect value $d_v$ or the outdegree $\beta_v$ of some node, but with a value $\gamma_v$ that is essentially equal to the ratio $\beta_v/(d_v+1)$ such that the list size of $v$ is proportional to $\gamma_v^2$. More formally, we partition the nodes into so-called $\gamma$-classes such that nodes in the same class have the same value $\gamma_v$ and hence a similar $\beta_v/(d_v+1)$ ratio. The details appear in the subsequent section.

\subsubsection{\boldmath\texorpdfstring{$\gamma$}{ Gamma}-Classes and Parameters}
\label{sec:gammaclasses}
Each node $v$ comes with some parameter $\gamma_v = 2^i$ for some $i \in [h]$, where $h$ is a parameter. We call $i$ the $\gamma$-class of $v$. Since these $\gamma$-classes have a natural order, we call a node $u$ to be in a \textit{lower} (respectively \textit{higher}) $\gamma$-class than $v$ if $i_u < i_v$ (respectively $i_u > i_v$). We define the following two parameters, which both depend on the maximum $\gamma$-class index $h$, the color space $\calC$, and the initial (proper) $m$-coloring of $G$.
\begin{align}\label{eq:tau}
  \tau(h, \calC, m) & := \lceil 8h + 2\log\log |\calC| + 2\log\log m + 16 \rceil,\\
  \tau'(h, \calC, m) & := 2^{\tau(h, \calC, m) - \lceil 2h + \log (2e) \rceil}\label{eq:tauprime}.
\end{align}
Note that these choices are consistent with the parameter setting in \cref{lemma:InterfaceToYannic}. If clear from the context, we omit the parameters $h, \calC$ and $m$ for simplicity and denote $\tau(h, \calC, m)$ by $\tau$ and $\tau'(h, \calC, m)$ by $\tau'$. \par

A node $v$ of $\gamma$-class $i_v$ is equipped with a color list $L_v$ of size $|L_v| = \ell_{i_v} := \alpha \cdot 4^{i_v} \tau (2g+1) = \alpha \gamma_v^2 \tau (2g+1)$ for some sufficiently large $\alpha > 0$ and $g \geq 0$. \par 
Because we solve a generalized version of the OLDC problem, we have to use a more general form for our conflict relation $\Psi$. Let $x \in \calC$ be a color and $C \subseteq \calC$ a set of colors, we denote the number of conflicts of $x$ with colors in $C$ regarding some given $g \geq 0$ by $\mu_g(x, C) := |\{c \in C \ | \ |x - c| \leq g \}|$. 

\begin{definition}[$\tau \& g$-conflict]
  \label{def:tauAndg}
  Two lists $C, C' \subseteq \mathcal{C}$ do  $\tau \& g$-conflict if $\sum_{x \in C} \mu_g(x, C') \geq \tau$.
\end{definition}

Note that $\sum_{x \in C} \mu_g(x, C') = \sum_{x \in C'} \mu_g(x, C)$ is always true. The $\Psi$ conflict relation from \Cref{def:PsiConflict} is adapted accordingly.

\begin{definition}[Conflict relation $\Psi_g(\tau',\tau)$]
  \label{def:PsiConflict2}
  Let $\tau',\tau>0$ be two parameters. The relation $\Psi_g(\tau',\tau)\subseteq 2^{2^{\calC}}\times 2^{2^{\calC}}$ is defined as follows. For any $K_1,K_2\in 2^{2^{\calC}}$, we have
  \begin{eqnarray*}
    (K_1,K_2) \in \Psi_g(\tau',\tau) \Leftrightarrow \exists\text{ distinct } C_1,\dots,C_{\tau'}\in K_1\ \text{s.t.}\ \\
                                     \forall i\in \set{1,\dots,\tau'}\,\exists\, C\in K_2 \text{ for which } C_i \text{ and } C \text{ do } \tau \& g \text{-conflict}.
  \end{eqnarray*}
\end{definition}
  
We can now adapt the definitions of problem $P_1$ and $P_2$ to what we need for the generalized OLDC problem. We define $k_i := 2^{i} \cdot \tau$ and $k' := 2^{h} \cdot \tau'$. Subsequently, we denote the $\gamma$-class of a node $v$ by $i_v$.
\newpage
\begin{definition}[Problems $P_1$ and $P_2$]\label{def:P1P2}\ 
  \begin{itemize}
  \item[\bf\boldmath$P_1$:] Every node $v$ has to output $C_v \subseteq L_v$ of size $|C_v| = k_{i_v}$ s.t.\ there are at most $d_v/2$ outneighbors $u$ of $v$ s.t.\ $u$ is in $\gamma$-class $i_u \leq i_v$, $C_u$ and $C_v$ do $\tau \& g$-conflict.
  \item[\bf\boldmath$P_2$:] Every node $v$ has to output a list $K_v \in 2^{\binom{L_v}{k_{i_v}}}$ of size $|K_v| = k'$ s.t.\ for each outneighbor $u$ in $\gamma$-class $i_u \leq i_v$, $(K_v, K_u) \not \in \Psi_g(\tau', \tau)$.
  \end{itemize}
\end{definition}

\subsubsection{Zero-Round Solution}

In this section we show that problem $P_2$ can be solved without communication. The high-level idea of the first step is to adapt \cref{lemma:InterfaceToYannic} such that we can apply it even if the sizes of the initial color lists differ (see \cref{lemma:InterfaceWithClasses}). We start by using a simple trick to make sure that the color list $L_v$ of each node $v$ does not contain colors that are \textit{close} to each other i.e., there are no distinct colors $x_1, x_2 \in L_v$ s.t.  $|x_1 - x_2| \leq g$. After doing this, the conflict relation $\Psi_g$ behaves almost the same as $\Psi$ in the fundamental problem. \par

We restrict color list $L_v$ to contain only colors that are from the same congruence class modulo $2g+1$ (recall that colors are from $\mathbb{N}_0$). In this way, for any two color lists $L_v$ and $L_u$ of colors, each color in $L_v$ can only $\tau \& g$-conflict with a single color in $L_u$. Formally, for every subset $P\subseteq \calC$ of the colors and every $a\in \mathbb{Z}_{2g+1}$, we define
\[
  \forall a\in \mathbb{Z}_{2g+1}\,:\,P^a:=\set{x \in P: x \equiv a \pmod{2g+1}}.
\]
For every list of colors $L$, we further define
\begin{align*}
  S_i^a(L) := \stackedbinom{L^a}{k_i}{k'}\text{ and }
  S_i(L) := \bigcup_{a\in \mathbb{Z}_{2g+1}} S_i^a(L).
\end{align*}  
We also define $\ell_i' := \ell_i/(2g+1)$ and
\begin{align*}
  \mathcal{L}_{i, a} :=  \binom{\calC^a}{\ell_i'} \text{ and }\mathcal{L}_{i} := \bigcup_{0 \leq a < 2g+1} \mathcal{L}_{i, a}.
\end{align*}

\begin{lemma}
    \label{lemma:InterfaceWithClasses}
    Let $i$ be a $\gamma$-class and let $L\in \mathcal{L}_{i}$. 
    Then, there exists a $\bar{S}_i(L)\subseteq S_i(L)$ such that $|\bar{S}_i(L)|\geq |S_i(L)|/2$ and such that for every $K\in \bar{S}_i(L)$ and every $L'\in \mathcal{L}_{i}$, there are at most $d_2(i)<\frac{1}{4m|\calC|^{\ell_i'}}\cdot|S_i(L)|$ different $K' \in S_i(L')$ such that $(K, K') \in \Psi_g(\tau', \tau)$ or $(K', K) \in \Psi_g(\tau', \tau)$. 
    Further, let $j \leq i$ and $L'\in \mathcal{L}_{j}$. For all $K' \in S_j(L')$, there are at most $d_2(i)$ $K\in S_i(L)$ such that $(K, K') \in \Psi_g(\tau', \tau)$. 
\end{lemma}
\begin{proof}
  This lemma relies on \cref{lemma:InterfaceToYannic}. However, the conflict relation $\Psi$ in \cref{lemma:InterfaceToYannic} is a special case of $\Psi_g$. Thus, we will show that whenever we consider color lists that exclusively contain colors of the same residual class, we can not construct more conflicts regarding $\Psi_g$ as we can in the setting of \cref{lemma:InterfaceToYannic} regarding $\Psi$.
  \begin{claim}
      For any $i, j$, let $L_1 \in \mathcal{L}_i$ and $L_2 \in \mathcal{L}_j$ be some arbitrary color lists.
      Let $C_1 \subseteq L_1^a$ and $C_2 \subseteq L_2^b$ for some $0 \leq a, b \leq 2g+1$ and let $[C_1]_b$ be a set that contains the colors of $C_1$ rounded to the closest value congruent to $b \mod 2g+1$. $C_1$ and $C_2$ do $\tau \& g$-conflict iff $[C_1]_b$ and $C_2$ do $\tau \& 0$-conflict  (i.e., $|[C_1]_b \cap C_2| \geq \tau$). 
  \end{claim}
  \begin{proof} \renewcommand{\qedsymbol}{$\blacksquare$}
      For all distinct colors $x, x' \in C_1$ we have $|x - x'| > 2g+1$. Thus, $x$ and $x'$ will be rounded to different values in $[C_1]_b$ which implies $|C_1| = |[C_1]_b|$. Let $x_1 \in C_1$, $x_2 \in C_2$ and $x_1' \in [C_1]_b$ be the rounded value of $x_1$. If $|x_1 - x_2| \leq g$, then there is no value in the residual class of $b$ that is closer to $x_1$ than $x_2$ is and hence $x_1' = x_2$. If on the other hand $|x_1 - x_2| > g$, there is a closer value than $x_2$ and hence $x_1' \not = x_2$. By construction, color $x_1$ conflicts with at most $1$ color from $C_2$ (otherwise not all elements in $C_2$ would be congruent to $b$).
      Hence, the statement of the lemma follows from $\sum_{x \in C_2} \mu_g(x, C_1) = \sum_{x \in C_2} \mu_0(x, [C_1]_b)$. 
  \end{proof}
  \begin{claim}
      \label{claim:equivalenceOfPsi}
      For any $L \subseteq \binom{\calC}{l_i}$ let $K \in S_i^a(L)$ and $K' \in S_i^b(L)$. For $K = \{C_1, ..., C_{k'}\}$ we define $[K]_b := \{[C_1]_b, \dots, [C_{k'}]_{b}\}$. If $(K, K') \in \Psi_g(\tau', \tau)$ then $([K]_b, K') \in \Psi(\tau', \tau)$. If $(K', K) \in \Psi_g(\tau', \tau)$ then $(K', [K]_b) \in \Psi(\tau', \tau)$. 
  \end{claim}
  \begin{proof}\renewcommand{\qedsymbol}{$\blacksquare$}
      If $(K, K') \in \Psi_g(\tau', \tau)$ there exists the distinct sets $C_1, \dots, C_{\tau'}\in K$ and for all $0 \leq i \leq \tau'$ there is a $C \in K'$ s.t. $C$ and $C_i$ do $\tau \& g$-conflict. Combined with above's claim we have that for each set $[C_1]_b, \dots, [C_{\tau'}]_b$ there exists $C \in K'$ s.t. $|C \cap [C_i]_b| \geq \tau$. By definition of $[K]_b$ we have $[C_1]_b, \dots, [C_{\tau'}]_b \in [K]_b$ and hence $([K]_b, K') \in \Psi(\tau', \tau)$. The second implication follows by the same argument.
  \end{proof}
  The following lines go along the proof of \cref{lemma:InterfaceToYannic} (since we fix some $\gamma$-class $i$ we basically have the same setting). We start by defining the conflict degree $d_2(i)$ in the same way as we did there:
  \begin{align*}
    d_1(i) := \binom{k_i}{\tau} \binom{\ell_i' - \tau}{k_i - \tau}, \quad
    d_2(i) := 4 \cdot \binom{k' \cdot d_1(i)}{\tau'} \cdot \binom{\binom{\ell_i'}{k_i} - \tau'}{k' - \tau'}
  \end{align*}
  By \cref{claim:equivalenceOfPsi}, given $K \in S_i(L)$ the number of different $K' \in S_i^b(L')$ for some arbitrary $0 \leq b \leq (2g+1)$ s.t. $(K, K') \in \Psi_g(\tau', \tau)$ is at most the number of different $K' \in S_i^b(L')$ s.t. $([K]_b, K') \in \Psi(\tau', \tau)$ and thus is upper bounded by $d_2(i)/4$ (due to $|L| = |L'| = l_i'$ we can go along the lines of \cref{sec:adaptionOf}). Defining $\bar{S}_i(L)$ s.t. for each element $K \in \bar{S}_i(L)$ there are at most $d_2(i)$ many $K' \in S_i^b(L')$ with $(K, K') \in \Psi_g(\tau', \tau)$ or $(K', K)\in \Psi_g(\tau', \tau)$, by \cref{claim:GoodSet} we have $|\bar{S}_i(L)| \geq |S_i(L)|/2$. Furthermore, for each $K \in \bar{S}_i(L)$ the number of different $K' \in S_i^b(L')$ with $(K, K') \in \Psi_g(\tau', \tau)$ or $(K', K) \in \Psi_g(\tau', \tau)$ is at most the number of $K' \in S_i^b(L')$ s.t. $([K]_b, K') \in \Psi(\tau', \tau)$ or $(K', [K_b]) \in \Psi(\tau', \tau)$ and thus at most $d_2(i)$. \par
  We will now show the last part of the Lemma: Let $j \leq i$ and $L'\in \mathcal{L}_{j}$. For each $K' \in S_j(L')$, there are at most $d_2(i)$ $K\in S_i^b(L)$ such that $(K, K') \in \Psi_g(\tau', \tau)$. Note that for $i=j$ this statement already follows from the above analysis. Let $L''$ be a superset of $L'$ that is filled up with arbitrary colors $x \equiv b \mod 2g+1$ such that $|L''| = |L|$. By that, $L''$ and $L$ imply the same $\gamma$-class $i$ and hence we have that there are at most $d_2(i)$ many $K'' \in S_i(L'')$ with $(K, K'') \in \Psi_g(\tau', \tau)$. By construction, for each $K'\in S_j(L')$ there exists a $K'' \in S_i(L'')$ that contains all the color lists in $K'$ but extends them by some arbitrary colors.  
  So if $(K, K') \in \Psi_g(\tau', \tau)$ then $(K, K'')\in \Psi_g(\tau', \tau)$. Thus, there are at most as many sets $K'\in S_j(L')$ that are in conflict with $K$ than $K''\in S_i(L'')$ that are in conflict with $K$, which ends the proof. 
\end{proof}

The next lemma shows the $0$-round solvability of $P_2$. Note that the ideas we use here are based on those of the $0$-round solution of \cref{sec:overview}. 

\begin{lemma}
    \label{lemma:SolvingNewP2}
    $P_2$ can be solved without communication given an initial $m$-coloring.
\end{lemma}
\begin{proof}
  Each node $v$ computes a value $a_v$ satisfying $0 \leq a_v < 2g+1$ that maximizes the size of $L_v^{a_v}$. By the pigeonhole principle $|L_v^{a_v}| \geq |L_v|/(2g+1) = \ell_{i_v}'$ where $i_v$ is the $\gamma$-class of $v$. Colors in $L_v \setminus L_v^{a_v}$ will be ignored in the subsequent steps. Note that this is necessary to apply \cref{lemma:InterfaceWithClasses}. \par
  We define the type $T_v$ of a node as the tuple $(c, L_v^{a_v})$ where $c$ is the color of $v$ in an initial $m$-coloring. Further, we can assign each type a $\gamma$-class since the $\gamma$-class depends on the corresponding list size (note that the $\gamma$-class of a node can be determined by the size of its color list. In \cref{sec:algOLDCSingleDef} this is formulated in more detail). \par
  Let $T_1, \dots, T_t$ be a fixed ordering of the types where we require that for any types $T_i = (c, L)$ and $T_j = (c', L')$ with $j < i$ we have $|L'| \leq |L|$. The number of types $t$ can be upper bounded by
  \begin{eqnarray*}
    t & = & m \sum_{i=1}^h \sum_{a=0}^{2g} |\mathcal{L}_{i, a}| \ \leq\ m \sum_{i=1}^h \binom{|\mathcal{C}|}{l_i'}\\
    & \leq & m \left(|\mathcal{C}|^{l_h'} + \sum_{i=1}^{h-1} |\mathcal{C}|^{l_i'} \right) \ \leq\  m \left(|\mathcal{C}|^{l_h'} + \frac{h-1}{|\mathcal{C}|^{3l_{h-1}'}}|\mathcal{C}|^{l_h'} \right) \ <\  2m|\calC|^{l_h'},
  \end{eqnarray*}
  where $h$ is the largest $\gamma$-class. We assign sets $K_i$ greedily to types i.e., for a given type $T_i = (c, L)$ let $1 \leq r \leq h$ be the $\gamma$-class implied by $L$, we choose some $K_i \in \bar{S}_r(L)$ (guaranteed to exist by \cref{lemma:InterfaceWithClasses}). Assume for any $i \geq 1$, each type $T_j$ for $j \in [i-1]$ already picked a $K_j$. Then $T_i = (c, L)$ has to pick some $K_i$ that does not conflict with any of the $i-1$ previous types, i.e., $(K_i, K_j) \not \in \Psi_g(\tau', \tau)$ for all $j \in [i-1]$. By \cref{lemma:InterfaceWithClasses}, there are at most $d(r) < \frac{|S_r(L)|}{4m|\mathcal{C}|^{l_r'}}$ conflicts to any type $T_j$ with $j < i$. Thus, the maximum number of conflicts of $T_i$ is $m \sum_{z=1}^r \sum_{a=0}^{2g} |\mathcal{L}_{z, a}| d_2(r) \leq 2m|\calC|^{l_r'}d_2(r) < |S_r(L)|/2 \leq |\bar{S}_r(L)|$. Hence, there is some conflict-free $K_i \in \bar{S}_r(L)$ that can be picked by $T_i$. By that, $P_2$ can be solved in $0$ rounds.
\end{proof}

\subsubsection{Algorithm} 
\label{sec:algOLDCSingleDef}
Assume each node is equipped with a color list of size $|L_v| \geq \alpha \left(\beta_v/(d_v+1) \right)^2 \tau (2g+1)$ and some defect value $d_v > 0$. The $\gamma$-class of a node $v$ is defined as the smallest $i_v$ s.t. $2^{i_v} \geq \frac{2\beta_v}{d_v+1}$. Based on its individual $\gamma$-class, each node solves $P_2$ (\cref{lemma:SolvingNewP2}) and forwards the solution to the neighbors. The knowledge gained by that is used to solve $P_1$ without additional communication. In more detail, node $v$ comes with list $K_v$ s.t. $(K_v, K_u) \not \in \Psi_g(\tau', \tau)$ for any outneighbor $u$ with $i_u \leq i_v$. This implies that at most $\tau'-1$ lists $C \in K_v$ do $\tau \& g$-conflict with some list in $K_u$. Hence, there are at most $\beta_v(\tau'-1)$ many $C \in K_v$ that $\tau \& g$-conflict. By the pigeonhole principle there is some $C_v \in K_v$ with at most\footnote{By definition of the $\gamma$-class, we have $(d_v + 1)/2 \geq \beta_v/2^h$} $\beta_v(\tau'-1)/k' < (d_v+1)/2$ many such conflicts within all the $C_u$ of outneighbors $u$ of smaller $\gamma$-classes. Hence, $C_v$ is a valid solution for $P_1$ that is then forwarded to its neighbors. \par

To solve the list coloring problem itself we iterate through the $\gamma$-classes in descending order. Each node $v$ has to decide on a color $x \in C_v$ s.t.\ in the end at most $d_v$ outneighbors are colored with the same color $x$. Let us fix a node $v$. By design, in the iteration $v$ decides on a color, all outneighbors of higher $\gamma$-classes already decided on a color and $v$ knows the $P_1$ solution lists $C_u$ of the outneighbors $u$ with $i_u \leq i_v$. Let $f_v(x)$ be the frequency of color $x$ within the outneighbors $u$ of $v$ i.e., the sum over all occurrences of $x$ in the $C_u$ sets of neighbors with the same or smaller $\gamma$-class plus the number of outneighbors of higher $\gamma$-classes that are already colored with $x$. The color $x$ with the lowest frequency in $C_v$ will be the final color of $v$ for the following reason: There are at most $d_v/2$ outneighbors that have an unbounded number of $\tau \& g$-conflicts, while $C_v$ shares at most $\tau-1$ colors with the remaining $C_u$ sets (note that with outneighbors of higher $\gamma$-class at most one color in $C_v$ is in conflict, hence, for worst-case observation we can ignore that case). By the pigeonhole principle there exists a color $x \in C_v$ with
\begin{align*}
    f_v(x) \leq \frac{\sum_{c \in C_v} f_v(c)}{|C_v|} \leq \frac{d_v/2 \cdot |C_v| + \beta_v (\tau-1)}{|C_v|} < d_v + 1.
\end{align*}
The round complexity of the whole algorithm is $O(h)$, since we iterate through all the $\gamma$-classes to assign colors. This completes the algorithm to handle \textit{single}-defects. We will now extend this result to solve the OLDC problem.

  \subsubsection{Multiple Defects}
  \begin{lemma}
      \label{lemma:SimpleOLDC}
      Given a graph with an initial $m$-coloring, there is an $O(h)$-round \LOCAL algorithm that assigns each node $v$ a color $x_v\in L_v$ such that every node $v\in V$ has at most $d_v(x_v)$ outneighbors $w$ with a color $x_w$ for which $|x_w-x_v|\leq g$ if for each node $v \in V$
      \begin{align*}
          \sum_{x \in L_v} (d_v(x)+1)^2 \geq \alpha \beta_v^2 \cdot \tau(h, \calC, m) \cdot h \cdot (2g+1)
      \end{align*}
      for some sufficiently large constant $\alpha$, some integer $h \geq \max_v \lceil \log(\frac{\beta_v}{\min_{x \in L_v}d_v(x)+1}) \rceil$ and color space $\calC$. Messages are of size at most $O(\min\{ \Lambda \cdot \log |\calC| , |\calC| \} + \log \log \beta + \log m)$-bits.
  \end{lemma}
  \begin{proof}
      Assume in the following that the values $d_v(x) + 1$ as well as $\beta_v$ are powers of $2$ for all nodes and colors. Note that this can be enforced by rounding. \par
      For each $1 \leq i \leq h$, let the set $L_{v, i} := \{x \in L_v ~|~ i =  \log_2(\beta_v/(d_v(x)+1)) \}$ be a subset of colors of $L_v$ where all colors share the same defect value.
      \begin{align*}
          \sum_{x \in L_v} (d_v(x) + 1)^2 = \sum_{i=1}^h \sum_{x \in L_{v, i}} (d_v(x)+1)^2 
      \end{align*}
      Hence, there exists some $i^*$ with the following property. Let $d_v$ be the defect of the colors in  $L_{v, i^*}$.
      \begin{align*}
          |L_{v, i^*}| \cdot (d_v+1)^2 = \sum_{x \in L_{v, i^*}} (d_v(x) + 1)^2 \geq \frac{\sum_{x \in L_v} (d_v(x) + 1)^2}{h} \geq \alpha \beta_v^2 \cdot \tau(h, \calC, m) \cdot (2g+1)
      \end{align*}
      Thus, restricting the set of colors to $L_{v, i^*}$ the algorithm described in \cref{sec:algOLDCSingleDef} solves the required coloring problem with (single)defect $d_v$ in $O(h)$ rounds. \par
      Now we will focus on the message complexity. To solve $P_1$, nodes have to communicate their $K_v$ list with the neighbors. Since $K_v$ is uniquely determined by $d_v$, $L_{v, i^*}$ and the initial color, we forward these values instead of $K_v$. By $|L_{v, i^*}| \leq |L_v| \leq \Lambda$, we need $O(\min \{ \Lambda \cdot \log |\calC|, |\calC| \})$-bits to send $L_{v, i^*}$ since we can either send a bit string of $|\calC|$ entries, indicating what colors of the globally known color space $\cal{C}$ are in $L_{v, i^*}$ or by simply sending the colors, where each color has size at most $O(\log |\calC|)$. Sending the defect naively takes $O(\log \beta)$-bits. Note that the defects are a power of $2$, hence $O(\log \log \beta)$ bits are indeed sufficient to send.
      Thus, together we need to transmit $O(\min \{ \Lambda \cdot \log |\calC|, |\calC| \} + \log \log \beta + \log m)$-bits. To solve the coloring, the set $C_v$ has to be forwarded. We encode $C_v$ as an index in $K_v$, thus, at most $O(\log k') = O(\Lambda)$ bits have to be sent. Thus, no message needs more bits than $O(\min \{ \Lambda \cdot \log |\calC|, |\calC| \} + \log \log \beta + \log m)$. 
  \end{proof}

  To apply this lemma, one can use the fact that $h = O(\log \beta)$ and for some color space of size $\poly \Delta$ and initial $O(\Delta^2)$-coloring (e.g., by \cite{Linial1987}) we have $\tau(h, \calC, m) = O(\log \Delta)$. Hence, \cref{lemma:SimpleOLDC} solves the OLDC problem (where $g=0$) in $O(\log \Delta)$ communication rounds if the initial color list $L_v$ for each node $v$ fulfills the condition $\sum_{x \in L_v} (d_v(x)+1)^2 \geq \alpha \beta_v^2 \cdot \log^2 \Delta$ for some sufficiently large constant $\alpha$. 
  Note that the OLDC algorithm mentioned in \cref{thm:def_local_list_coloring} requires a stronger condition on the color lists $L_v$ than \cref{lemma:SimpleOLDC}. However, to reach this better result, we apply \cref{lemma:SimpleOLDC} as a subroutine. The details of this more involved analysis are stated in the next section. 


\subsection{Main Contribution}

If we have an OLDC instance in which some nodes have colors with constant defect requirements, the number $h$ of $\gamma$-classes can be $\Theta(\beta)$. Because also the value of $\tau(h,|\calC|,m)$ is linear in $h$, this means that even if $g=0$ and even if $|\calC|$ and $m$ are both polynomial in $\beta$, the condition in \Cref{lemma:SimpleOLDC} is of the form $\sum_{x\in L_v}(d_v(x)+1)^2\geq \alpha\beta_v^2\log^2\beta$. One of the $\log\beta$ factors comes from the fact that at the very beginning of the algorithm, every node $v$ reduces its color list to a list in which all colors have approximately the same defect value. In the following, we show that at the cost of a more complicated algorithm, we can improve this $\log\beta$ factor to a factor of the form $\poly\log\log\beta$. In the following discussion, we assume that $g=0$, but we note that along the way, we will have to use \Cref{lemma:SimpleOLDC} with positive $g$ as a subroutine.

In order to obtain the improvement, we first want an algorithm where for computing the $0$-round problem $P_2$, a node $v$ of some $\gamma$-class only needs to compete with outneighbors of the same $\gamma$-class. For this, we use an iterative approach to solve $P_2$ and $P_1$. For each $i\in [h]$, let $V_i\subseteq V$ be the set of nodes in $\gamma$-class $i$. For each node $v\in V_i$ that is colored with some color $x$, we will make sure that $v$ has at most $d_v(x)/4$ outneighbors of color $x$ in $\gamma$-classes $j$ for $j<i$, at most $d_v(x)/4$ outneighbors of color $x$ in the same $\gamma$-class and that $v$ has at most $d_v/2$ outneighbors in $\gamma$-classes $j$ for $j>i$. Thus, we assume that each node $v\in V$ only uses the part $L_{v,i}$ of its $L_v$ consisting of colors with a defect $d_v$ such that $\gamma_v=2^i \geq 4\beta_v /(d_v+1)$. We then iterate over the $\gamma$ classes $i\in [h]$ in increasing order. In iteration $i$, we solve problems $P_2$ and $P_1$ for the nodes in $V_i$. When dealing with nodes in $V_i$, we can therefore assume that for all outneighbors in $u\in V_j$ for $j<i$, the list $C_u$ (i.e., the output of problem $P_1$) is already computed. We then remove each color $x$ from the list $L_{v,i}$ for which there are more than $d_v(x)/4$ outneighbors $u\in V_1\cup\dots\cup V_{i-1}$ for which $x\in C_u$. In this way, we guarantee that $v$ cannot choose a color with defect more than $d_v/4$ to outneighbors in lower $\gamma$-classes even before solving $P_2$ for node $v$. We can then solve $P_2$ and $P_1$ by only considering outneighbors in $V_i$. Of course, we have to make sure that even after removing colors from $L_{v,i}$, the list of $v$ is still sufficiently large to solve problem $P_2$. The advantage of only having to consider neighbors in $V_i$ when solving $P_2$ and $P_1$ is that in the condition on $\sum_{x\in L_{v,i}} (d_v(x)+1)^2$, we can replace the outdegree $\beta_v$ of $v$ by the number of outneighbors $\beta_{v,i}$ that $v$ has in $V_i$. This gives us more flexibility in the choice of $v$'s $\gamma$-class $i$. If $\sum_{x\in L_{v,i}} (d_v(x)+1)^2$ is large, $v$ can choose $\gamma$-class $i$ and tolerate many outneighbors in the same $\gamma$-class and if $\sum_{x\in L_{v,i}} (d_v(x)+1)^2$ is small, $v$ can only choose $\gamma$-class $i$ if a small number of outneighbors choose $\gamma$-class $i$. We will see that the problem of choosing a good $\gamma$-class can be phrased as an OLDC problem that can be solved by using \Cref{lemma:SimpleOLDC} with appropriate parameters. The following technical \cref{lemma:technicalMain} assumes that the $\gamma$-classes are already assigned and it formally proves under which conditions the above algorithmic idea allows us to solve a given OLDC instance.

\begin{lemma}\label{lemma:technicalMain}
  Let $G=(V,E)$ be a directed graph that is equipped with an initial proper $m$-coloring. Let $h\geq 1$ be an integer parameter and assume that every node $v\in V$ is in some $\gamma$-class $i_v\in [h]$. For every $i\in [h]$, let $V_i$ be the nodes in $\gamma$-class $i$ and let $\beta_{v,i}$ be the number of outneighbors of $v$ in $V_i$. Each node $v$ has a color list $L_v\subseteq \calC$ and one fixed defect value $d_v$, i.e., $d_v(x)=d_v$ for all $x\in L_v$. We further define $\tau :=\tau(h, \calC, m)$ and some integer parameter $q\in [\tau]$. We assume that for all $v \in V$
  \[
    \forall v\in V\,:\, \frac{4 \cdot \max\set{\beta_{v,i_v}, \frac{\beta_v}{q}}}{d_v+1} \leq 2^{i_v} \quad\text{and}\quad 
    |L_v| \geq \left[ \alpha \cdot 4^{i_v} +
    \frac{4}{d_v+1}\cdot\sum_{j=i_v-\lfloor \log q\rfloor}^{i_v-1} \beta_{v,j}\cdot 2^j\right]\cdot \tau
  \]
  for a sufficiently large constant $\alpha>0$. Then there is an $O(h)$-round algorithm that assigns each node $v$ a color $x\in L_v$ such that every node $v\in V$ has at most $d_v$ outneighbors of color $x$. The algorithm requires messages consisting of $O(\min\{\Lambda \log |\mathcal{C}|, |\mathcal{C}| \} + \log \log \beta + \log m)$ bits.
\end{lemma}
\begin{proof}
  As discussed, the algorithm consists of two phases I and II. In Phase I, we iterate over the $\gamma$-classes $1,\dots,h$ in increasing order. The objective of Phase I is for every node $v\in V$ a modified version of $P_1$. Specifically, at the end of Phase I, every node $v$ has to output a set $C_v\subseteq L_v$ such that if $v\in V_i$, $|C_v|=2^i\cdot \tau$ and for every $x\in C_v$, there are at most $d_v/4$ outneighbors $u\in V_j$ for $j<i$ for which $x\in C_u$ and such that for all except at most $d_v/4$ outneighbors $u\in V_i$, we have $|C_v\cap C_u|< \tau$. In Phase II, we iterate over the $\gamma$-classes in decreasing order and we use the sets $C_v$ computed in Phase I to compute the final solution to the given OLDC instance.

  \para{Phase I.} Consider iteration $i\in [h]$ in Phase I, when we process the nodes in $V_i$. At this point, we have already computed the sets $C_u$ for all nodes $u\in V_j$ for $j<i$. Every node $v\in V_i$ proceeds as follows. We define the bad colors $B_v\subseteq L_v$ for $v$ as the set of colors $x\in L_v$ for which there are more than $d_v/4$ outneighbors $u\in V_1\cup\dots\cup V_{i-1}$ for which $x\in C_u$. As a first step, node $v$ updates its list $L_v$ to $L_v' = L_v \setminus B_v$. We next lower bound the size of $L_v'$. For this, we define
  \begin{eqnarray*}
    D & :=& \sum_{j=1}^{i-1}\sum_{u\in V_j\cap \Nout(v)} |C_u|\ \ = \ \
            \sum_{j=1}^{i-1} 2^j\cdot \beta_{v,j}\cdot \tau\\
      & \leq & \sum_{j=1}^{i-\lfloor\log q\rfloor -1}2^j\cdot \beta_{v}\cdot\tau + \sum_{j=i - \lfloor\log q\rfloor}^{i-1} 2^j\cdot \beta_{v,j}\cdot \tau\\
     &  < & 2^{i+1}\cdot \frac{\beta_v}{q} \cdot\tau + \sum_{j=i - \lfloor\log q\rfloor}^{i-1} 2^j\cdot \beta_{v,j}\cdot \tau.
  \end{eqnarray*}
  The smallest integer larger than $d_v/4$ has value at least $(d_v+1)/4$. We therefore have $|B_v| \leq 4D/(d_v+1)$. By using our upper bound on $D$ and the lower bound on $|L_v|$ given by the lemma statement, we can thus bound $|L_v'|$ as
  \begin{eqnarray*}
    |L_v'| =  |L_v| - |B_v|
    & \geq & \left[\alpha\cdot 4^{i} - \frac{2^{i+3}\cdot \beta_v}{q\cdot (d_v+1)}\right]\cdot\tau\\
    & \geq & \left[(\alpha-2)\cdot4^{i} + 2\cdot 2^i\cdot \frac{4\cdot \frac{\beta_v}{q}}{d_v+1} - \frac{2^{i+3}\cdot \beta_v}{q\cdot (d_v+1)}\right]\cdot\tau\\
    & = & (\alpha-2)\cdot 4^i \cdot\tau.
  \end{eqnarray*}
  The second inequality follows because $2^i\geq \frac{4\beta_v/q}{d_v+1}$ and $i\geq 1$. The goal of each node $v\in V_i$ now is to select a set $C_v$ such that there are at most $d_v/4$ outneighbors $u\in V_i$ for which $|C_u\cap C_v|\geq\tau$. For this, we first solve a special case of problem $P_2$ defined in \Cref{sec:gammaclasses} and which is identical to the problem $P_2$ solved in \cite{MausT20}. Every node $v\in V_i$ computes a set $K_v\subseteq 2^{\binom{L_v'}{2^i\tau}}$ of size $2^i\cdot\tau'\leq 2^h\cdot\tau'$, where $\tau':=\tau'(h, \calC, m)$. For every outneighbor $u$ of $v$, it is guaranteed that $(K_v,K_u)\not\in \Psi_0(\tau', \tau)$, i.e., $K_v$ contains at most $\tau'-1$ sets $C$ for which there exists $C'\in K_u$ with $|C\cap C'|\geq \tau$. By \Cref{lemma:SolvingNewP2} (or by Lemma 3.2 in \cite{MausT20}), those sets $K_v$ can be computed in 0 rounds. Each node $v\in V_i$ now chooses a set $C_v\in K_v$ for which there are at most $d_v/4$ outneighbors $u\in V_i$ for which $|C_u\cap C_v|\geq \tau$. For each $C\in K_v$ and each $u\in \Nout\cap V_i$, let $I_{C,u}=1$ if there exists $C'\in K_u$ for which $|C\cap C'|\geq \tau$ and let $I_{C,u}=0$ otherwise. We know that for each $u\in \Nout\cap V_i$, there are at most $\tau'-1$ sets $C\in K_v$ for which $I_{C,u}=1$. We further define $D_C$ as the number of nodes $u\in \Nout\cap V_i$ for which $I_{C,u}=1$. That is, if $v$ chooses set $C\in K_v$, there are $D_C$ outneigbors to which the conflict is potentially too large and which might therefore contribute to the defect of $v$. We have
  \[
    \sum_{C\in K_v} D_C = \sum_{C\in K_v}\sum_{u\in \Nout\cap V_i}  I_{C,u} =  \sum_{u\in \Nout\cap V_i} \sum_{C\in K_v} I_{C,u} \leq \beta_{v,i}\cdot (\tau'-1).
  \]
  On average over all $C\in K_v$, the value of $D_C$ is therefore at most $\frac{\beta_{v,i}(\tau'-1)}{|K_v|} \leq \frac{\beta_{v,i}(\tau'-1)}{2^i \tau'} < \frac{\beta_{v,i}}{2^i}$. By using that $2^i\geq \frac{4\beta_{v,i}}{d_v+1}$, we can upper bound the average $D_C$ value by $<\frac{d_v+1}{4}$. Because the smallest integer larger than $d_v/4$ has value at least $\frac{d_v+1}{4}$, there must be a set $C \in K_v$ for which $D_C\leq d_v/4$. Node $v$ chooses such a set as its set $C_v$. This concludes Phase I. In the following, we use $\Nout_{i,*}(v)$ to denote the outneighbors $u$ of $v$ in $V_i$ for which $|C_u\cap C_v|<\tau$. The remaining at most $d_v/4$ outneighbors in $\Nout(v)\cap V_i\setminus \Nout_{i,*}(v)$ are ignored in the rest of the algorithm. They contribute at most $d_v/4$ to the overall defect of $v$.  
  
  \para{Phase II.} In Phase II, we again iterate over the $h$ $\gamma$-classes, but this time in decreasing order starting with $\gamma$-class $h$. That is, when the nodes $V_i$ of some $\gamma$-class $i$ choose their color, we can assume that all the nodes in $V_j$ for $j>i$ have already chosen their colors. Consider again a node $v\in V_i$. In Phase I, node $v$ has computed a set $C_v$ of size $|C_v|=2^i\cdot \tau$ such that for every outneighbor $u\in \Nout_{i,*}(v)$, it holds that $|C_u\cap C_v|<\tau$. We need to guarantee a defect of $\leq d_v/2$ to the outneighbors in $\Nout_{i,*}(v)\cup\dots\cup V_h$. Node $v$ therefore needs to choose a color $x\in C_v$ that appears at most $d_v/2$ times among the already chosen colors of outneighbors in higher $\gamma$-classes or among any of the set $C_u$ of outneighbors in $\Nout_{i,*}(v)$. The number of outneighbors of higher $\gamma$-classes is $\leq\beta_v\leq 2^i(d_v+1)\tau/4$. Here, we use that $2^i\geq \frac{4\beta_v/q}{d_v+1}$ and that $q\leq \tau$. Using that also $2^i\geq \frac{4\beta_{v,i}}{d_v+1}$, we further have
  \[
    \sum_{u\in \Nout_{i,*}(v)} |C_u\cap C_v| < \beta_{v,i}\cdot\tau \leq \frac{2^i(d_v+1)\tau}{4}.
  \]
  Overall, the multiset consisting of all the colors of $C_v$ that can be picked by any outneighbor of $v$ in $V_j$ for $j> i$ or in $\Nout_{i,*}(v)$ is therefore less than $\frac{2^i(d_v+1)\tau}{2}$. Hence, because $|C_v|=2^i\tau$, on average, the colors in $C_v$ appear less than $\frac{d_v+1}{2}$ times in this multiset. This implies that there exists a color $x\in C_v$ that can be chosen by at most $d_v/2$ outneighbors $u$ in $V_j$ for $j> i$ or in $\Nout_{i,*}(v)$. The solution of Phase I already guarantees that every color in $C_v$ can be chosen by at most $d_v/2$ of the other outneighbors. Overall, the total number of outneighbors of $v$ that can pick the same color is therefore at most $d_v$ as required.

  In both phases we iterate over the $h$ $\gamma$-classes. To determine the message complexity, note that each node has to transmit the sets $K_v$ and $C_v$ to all neighbors. Here we can use the same encoding as in \cref{lemma:SimpleOLDC} and hence, the maximum message per round does not require more than $O(\min\{\Lambda \log |\mathcal{C}|, |\mathcal{C}| \} + \log \log \beta + \log m)$-bits. 
\end{proof}

We are now ready to prove \Cref{thm:def_local_list_coloring}, our main contribution. In the following $\hat{\beta}_v$ is the outdegree of $v$ rounded up to the next integer power of $2$. Further $\hat{\beta}:=\max_{v\in V} \hat{\beta}_v$. Note that for all $v$, we have $\hat{\beta}_v\leq 2\beta_v$. The following lemma is a rephrasing of the theorem. By adjusting the constant $\alpha$, \Cref{thm:def_local_list_coloring} follows from \Cref{lemma:mainresult} because for $h=\lceil\log\hat{\beta}\rceil$ and $h'=\lceil\log 4h\rceil$, $\tau(h, \calC, m)=O(\log \beta + \log\log|\calC| + \log\log m)$ and $\tau(h', [h], m)=O(\log\log\beta + \log\log m)$.

\begin{lemma}\label{lemma:mainresult}
  Let $G=(V,E)$ be a properly $m$-colored directed graph and let $h:=\lceil\log\hat{\beta}\rceil$ and $h':=4^{\lceil\log_4\log 8h\rceil}$.  Assume that we are given an OLDC instance on $G$ for which
  \begin{equation}\label{eq:adaptedmaincond}
    \forall v\in V\,:\,\sum_{x\in L_v} \big(d_v(x)+1\big)^2 \geq \alpha^2\cdot \hat{\beta}_v^2 \cdot \tau\cdot\bar{\tau}\cdot h'^2,
  \end{equation}
  where $\tau = 4^{\lceil \log_4 \tau(h,\calC,m)\rceil}$ and $\bar{\tau} = 4^{\lceil\log_4 \tau(h', [h], m)\rceil}$.
  Then, there is a deterministic distributed algorithm that solves this OLDC instance in $O(\log\beta)$ rounds using $O\big(\min\set{|\calC|, \Lambda\cdot \log |\calC|} + \log \beta + \log m\big)$-bit messages.
\end{lemma}
\begin{proof}
  First note that w.l.o.g., we can assume that for all $v$ and all $x\in L_v$, $(d_v(x)+1)^2$ and $\alpha$ are both integer powers of $4$. We can just round up $\alpha$ and round down $d_v(x)$ to the next value for which this is true. We then just need to choose the constant $\alpha$ slightly larger. With those assumptions, the right-hand side of \eqref{eq:adaptedmaincond} is then a integer power of $4$. For each node $v$, we define $R_v:=\alpha\cdot\hat{\beta}_v^2 \cdot\bar{\tau}\cdot h'^2$. For every $v\in V$ and every $x\in L_v$, we then have $\frac{R_v}{(d_v(x)+1)^2}= 4^{\mu}$ for some $\mu\in[h]$. We can therefore partition each list $L_v$ in to lists $L_v=L_{v,1}\cup\dots\cup L_{v,h}$ such that for all $\mu\in [h]$, $L_{v,\mu}$ consists of the colors $x\in L_v$ for which $\frac{R_v}{(d_v(x)+1)^2}= 4^{\mu}$.  The algorithm to solve the given OLDC instance consists of two phases. In the first phase, every node $v\in V$ chooses its $\gamma$-class, which is an integer $i_v\in [h]$. In the second phase, we then use \Cref{lemma:technicalMain} to solve the OLDC instance.

  We first discuss the objective of the first phase and we consider some node $v$. For every $\mu\in [h]$, we define $D_{v,\mu} := \sum_{x\in L_{v,\mu}} (d_v(x)+1)^2$ and $D_v:=\sum_{\mu=1}^h D_{v,\mu} = \sum_{x\in L_v} (d_v(x)+1)^2$.  For each $\mu\in [h]$, we further define $\lambda_{v, \mu}\in(0,1]$ as follows
  \[
    \lambda_{v,\mu} :=
    \begin{cases}
      0 & \text{ if } D_{v,\mu}/D_v < 1/(2h)\\
      4^{\lfloor \log_4(D_{v,\mu}/D_v)\rfloor} & \text{ otherwise.}
    \end{cases}
  \]
  In the next part, we make a case distinction and first assume that $\lambda_{v,\mu}<1/4$ for all $\mu$. The case when there is some $\mu$ with $\lambda_{v,\mu}\geq 1/4$ is a simple case that we discuss later.

  \para{Case I : \boldmath $\forall \mu \in [h]\,:\,\lambda_{v,\mu}<1/4$.} Note that the definition of $\lambda_{v,\mu}$ implies that for all $\mu$ where $\lambda_{v,\mu} \not= 0$, $\lambda_{v,\mu}> 1/(8h)$ and $\lambda_{v,\mu}=4^{-r_{v,\mu}}$ for some integer $r_{v,\mu}\in\set{0,\dots,\lceil \log_4 4h\rceil}$. Note also that the values of $D_{v,\mu}$ for which $\lambda_{v,\mu}=0$ sum up to at most $D_v/2$ and therefore
  \[
    \sum_{\mu=1}^h \lambda_{v,\mu} \geq \frac{1}{8}.
  \]
  For every $v\in V$, we next define a function $f_v:[h]\to[h]\cup\set{\bot}$ such that for all $\mu\in[h]$, $f_v(\mu) = \mu - r_{v,\mu} + 2$ if $\lambda_{v,\mu}>0$ and $f_v(\mu)=\bot$ otherwise. Note that $0<\lambda_{v,\mu}<1/4$ implies that $f_v(\mu)\leq h$. For every $\mu\in [h]$, we next also define a second function $i_v:[h]\to[h]\cup \set{\bot}$ as follows. For every $\mu$, we set $i_v(\mu) = f_v(\mu)$ if $f_v(\mu)=\bot$ or if $f_v(\mu)\geq 1$ and there is no $\mu'<\mu$ for which $f_v(\mu')=f_v(\mu)$. Otherwise, we set $i_v(\mu) = \bot$. Note that for any two $\mu,\mu'\in[h]$ with $\mu\neq \mu'$, we either have $i_v(\mu)=i_v(\mu')=\bot$ or we have $i_v(\mu)\neq i_v(\mu')$. We next show that
  \begin{equation}\label{eq:finallambdasum}
    \sum_{\mu\in [h] : i_v(\mu)\neq \bot} \lambda_{v,\mu} \geq \frac{2}{3}\cdot \sum_{\mu\in [h] : f_v(\mu)\neq \bot \land f_v(\mu) \geq 1} \lambda_{v,\mu} \geq
    \frac{2}{3}\cdot \left(\frac{1}{8} - \frac{1}{48}\right)\geq \frac{1}{20}.
  \end{equation}
  To see this, consider some $\mu$ for which $\lambda_{v,\mu}>0$. Note that for $f_v(\mu)<1$, we need $r_{v,\mu} \geq \mu + 2$ and therefore $\lambda_{v,\mu} \leq 4^{-\mu - 2}$. Thus, the sum over those $\lambda_{v,\mu}$ is at most $\frac{4}{3}\cdot 4^{-1-2} = 1/48$. It therefore remains to show that the sum over the $\lambda_{v,\mu}$ for which $f_v(\mu)\geq 1$, but $i_v(\mu)=\bot$ is at most a third the sum over the $\lambda_{v,\mu}$ for which $f_v(\mu)\geq 1$. We have $i_v(\mu)=\bot$ and $f_v(\mu)\geq 1$ iff there is a $\mu'<\mu$ for which $f_v(\mu')=f_v(\mu)$. For $f_v(\mu')=f_v(\mu)$, we need to have $\lambda_{v,\mu} = \lambda_{v,\mu'}\cdot 4^{\mu'-\mu}$.  Consider some value $z\geq 1$ for which there is a value $\mu_z$ with $f_v(\mu_z)=z$ and assume that $\mu_z$ is the smallest such value. The sum over all $\lambda_{\mu}$ for $\mu>\mu_z$ and $f_v(\mu)=f_v(\mu_z)$ is at most $\sum_{\mu=\mu_z+1}^\infty \lambda_{\mu_z} \cdot 4^{\mu_z - \mu} = \lambda_{\mu_z}/3$. This concludes the proof of Inequality \eqref{eq:finallambdasum}.

  In order to assign a $\gamma$-class $i_v$ to every node $v\in V$, we define another (generalized) OLDC instance. For this instance, the "color" list of node $v$ is $\calL_v=\set{i_v(\mu) : i_v(\mu)\neq \bot}$. For every $i\in \calL_v$, we define the inverse function $\mu_v(i)$ to be the value $\mu$ for which $i_v(\mu)=i$. For each color $i\in \calL_v$, we then define a defect $\delta_{v,i}$ as
  \[
    \delta_{v,i} := \left\lfloor \sqrt{\lambda_{v,\mu_v(i)}\cdot R_v}\right\rfloor.
  \]
  We further define $q:= h$ and $g := \lfloor\log h\rfloor$. We then want to find an assignment of values $i_v\in \calL_v$ to each node such that for every $v\in V$, the number of outneighbors $u$ for which $i_u \in [i_v - g, i_v]$ is at most $\delta_{v,i_v}$. We next show that such an assignment of $\gamma$-classes $i_v$ satisfies the requirement needed by \Cref{lemma:technicalMain} and we can therefore use it to efficiently solve the original OLDC instance. We afterwards show that the generalized OLDC instance to find the values $i_v$ satisfies the requirement of \Cref{lemma:SimpleOLDC}.

  Let us therefore assume that we have an assignment of $\gamma$-class $i_v$ to the nodes that solve the above generalized OLDC problem. For each $i\in [h]$, we again use $V_i$ to denote the set of nodes $v$ with $i_v=i$ and we assume $\beta_{v,i}$ is the number of outneighbors of $v$ in $V_i$. The fact that $v$ has at most $\delta_{v,i_v}$ outneighbors $u$ with $i_u\in [i_v-g, i_v]$ implies that $\delta_{v,i_v}\geq \beta_{v,j}$ for all $j\in [i_v-g, i_v]$. For all $v\in V$, we have
  \[
    \delta_{v,i_v} = \left\lfloor \sqrt{\lambda_{v,\mu_v(i)}\cdot R_v}\right\rfloor
    \stackrel{(\lambda_{v,\mu_v(i)}\geq 1/(8h))}\geq
    \left\lfloor\sqrt{\frac{R_v}{8h}}\right\rfloor
    =
    \left\lfloor\sqrt{\frac{\alpha\cdot\hat{\beta}_v^2 \cdot\bar{\tau}\cdot h'^2}{8h}}\right\rfloor
    \geq \frac{\hat{\beta}_v}{h}.
  \]
  The last inequality follows because $h, \tau, \bar{\tau}, h'\geq 1$ and if we choose $\alpha\geq 8$. For the following calculations, we define $\mu_v := \mu_v(i_v) = i_v + r_{v,\mu} -2$. Note that if $v$ chooses $\gamma$-class $i_v$, it uses the colors in $L_{v,\mu_v}$. All those colors have a defect $d_v$ such that $(d_v+1)^2=R_v / 4^{\mu_v}$. Using $q=h$, we therefore have
  \begin{eqnarray*}
    \frac{4 \cdot \max\set{\beta_{v,i_v}, \frac{\beta_v}{q}}}{d_v+1}
    & \leq & \frac{4 \cdot \delta_{v,i_v}}{d_v+1}\\
    & \leq & \sqrt{\frac{16\lambda_{v,\mu_v}\cdot R_v}{(d_v+1)^2}}\\
    & = & \sqrt{16\lambda_{v,\mu_v}\cdot 4^{\mu_v}}\\
    & = & \sqrt{4^2\cdot 4^{-r_{v,\mu_v}}\cdot 4^{i_v +r_{v,\mu_v}-2}}\ =\ 2^{i_v}.
  \end{eqnarray*}
  The first part of the requirement of \Cref{lemma:technicalMain} is therefore satisfied. For the second part, recall that $v$ uses the colors in $L_{v,\mu_v}$ and that $D_{v,\mu_v}=\sum_{x\in L_{v,\mu_v}}(d_v(x)+1)^2 = |L_{v,\mu_v}|\cdot (d_v+1)^2$, where $d_v$ is defined as before. We have
  \begin{equation}\label{eq:listbound1}
    |L_{v,\mu_v}| \geq \frac{\lambda_{v,\mu_v} D_v}{(d_v+1)^2} \geq
    \frac{\lambda_{v,\mu_v} \cdot \alpha \tau \cdot R_v}{(d_v+1)^2} =
    4^{-r_{v,\mu_v}}\cdot \alpha\tau\cdot 4^{\mu_v} = \frac{\alpha}{16}\cdot 4^{i_v}\cdot \tau.
  \end{equation}
  Before we continue, we switch to Case II.
  
  \para{Case II : \boldmath $\exists \mu \in [h]\,:\,\lambda_{v,\mu}\geq 1/4$.}
  Before looking at the problem of assigning the $\gamma$-classes, we have a look at the requirements for \Cref{lemma:technicalMain} in Case II, i.e., if there is a $\mu\in [h]$ for which $\lambda_{v,\mu}\geq 1/4$. Let $\mu_v$ be one such value $\mu$. In this case, we set $i_v = \mu_v$, $\calL_v=\set{i_v}$, and $\delta_{v,i_v}:=\lfloor \sqrt{R_v}/4\rfloor$. We then have $\delta_{v,i_v}=\big\lfloor\sqrt{\alpha\hat{\beta}_v^2\bar{\tau}h'^2/16}\big\rfloor\geq \beta_v$. The last inequality holds if $\alpha\geq 16$ because $\bar{\tau}$ and $h'$ are positive integers. We therefore clearly have $\delta_{v,i_v}\geq \max\set{\beta_{v,i_v}, \beta_v/q}$ and therefore
    \[
      \frac{4\cdot\max\set{\beta_{v,i_v}, \frac{\beta_v}{q}}}{d_v+1} \leq \frac{4\delta_{v,i_v}}{d_v+1}
      \leq \sqrt{\frac{16 R_v}{16(d_v+1)^2}} = 2^{\mu_v} = 2^{i_v}.
    \]
    Hence, the first part of the requirement of \Cref{lemma:technicalMain} also holds in Case II. Similarly to Case I, we can lower bound the size of the color list $L_{v,\mu_v}$ that is used by $v$:
  \begin{equation}
    |L_{v,\mu_v}| \geq \frac{\lambda_{v,\mu_v} D_v}{(d_v+1)^2} \geq
    \frac{\alpha \tau \cdot R_v}{4(d_v+1)^2} =
    \frac{\alpha}{4}\cdot 4^{i_v}\cdot \tau.
  \end{equation}
  The bound given by \eqref{eq:listbound1} therefore also holds in Case II.

  We now continue considering both cases together. It remains to show (to apply \Cref{lemma:technicalMain}) that 
  \begin{equation}\label{eq:listbound2}
    |L_{v,\mu_v}| \geq  \alpha' \cdot 4^{i_v}\cdot \tau +
    \frac{4}{d_v+1}\cdot\sum_{j=i_v-\lfloor \log q\rfloor}^{i_v-1} \beta_{v,j}\cdot 2^j\cdot \tau
  \end{equation}
  for some constant $\alpha'$ that can be chosen as large as needed by choosing the constant $\alpha$ sufficiently large. Note that we have $g=\lfloor \log q\rfloor$ and thus for $j\in [i_v-\lfloor\log q\rfloor, i_v]$, $\beta_{v,j}\leq \delta_{v,i_v}$. We therefore have
  \[
    \frac{4}{d_v+1}\cdot\sum_{j=i_v-\lfloor \log q\rfloor}^{i_v-1} \beta_{v,j}\cdot 2^j\cdot \tau \leq \frac{4\delta_{v,i_v}\tau}{d_v+1} \cdot\sum_{\ell=1}^{g} 2^{i_v - \ell} < \frac{4\delta_{v,i_v}}{d_v+1}\cdot 2^{i_v}\cdot \tau.
  \]
  We have already seen that $4\delta_{v,i_v}/(d_v+1)\leq 2^{i_v}$ and the bound in the above inequality can therefore be upper bounded by $4^{i_v}\tau$. For every constant $\alpha'>0$, we can therefore choose a constant $\alpha>0$ such that the bound in \eqref{eq:listbound2} is upper bounded by the bound in \eqref{eq:listbound1}. This shows that if $v$ is in Case I, node $v$ satisfies the requirements to apply \Cref{lemma:technicalMain}.
  
  We next also show that the assignment of $\gamma$-classes $i_v$ can be done by using the algorithm of \Cref{lemma:SimpleOLDC}. Recall that every node $v$ needs to pick an $i_v\in \calL_v$ such that the total number of outneighbors $u$ that pick $i_u\in[i_v-g, i_v]$ is at most $\delta_{v,i_v}$. To apply \Cref{lemma:SimpleOLDC}, we have to lower bound $\sum_{i\in \calL_v}(\delta_{v,i}+1)^2$. We have
  \begin{eqnarray*}
    \sum_{i\in \calL_v}(\delta_{v,i}+1)^2
    & \geq &
             \min\set{\frac{1}{16},\sum_{i\in\calL_v}\lambda_{v,\mu_v(i)}}\cdot R_v\\
    & \stackrel{\eqref{eq:finallambdasum}}{\geq} & \frac{1}{20}\cdot R_v\\
    & = & \frac{1}{20}\cdot\alpha\cdot\hat{\beta}_v^2\cdot\bar{\tau}\cdot h'^2.
  \end{eqnarray*}
  Note that we have $g=\lfloor \log h\rfloor$ and $h'\geq \log(8h)$. We therefore have $2h'\geq 2g+1$. By choosing a sufficiently large constant $\alpha$, the above inequality therefore implies that the requirements of \Cref{lemma:SimpleOLDC} are satisfied as long as the value of $\bar{\tau}$ is sufficiently large. We have $\bar{\tau}=\tau(h', [h], m)$. Note that the color space of the OLDC problem that we use to assign the values $i_v$ is $[h]$. Recall that for all $v$ and $\mu$, $\lambda_{v,\mu}=0$ or $\lambda_{v,\mu}\geq 1/(8h)$. For each node $v\in V$, we therefore have
  \[
    \min_{i\in \calL_v} \delta_{v,i} \geq \min\set{\frac{\sqrt{R_v}}{4}, \sqrt{\lambda_{v,\mu_v(i)} R_v}} \geq
    \min\set{\frac{1}{4}, \frac{1}{\sqrt{8h}}}\cdot\sqrt{R_v} \geq \frac{\sqrt{R_v}}{8h} \geq \frac{\beta_v}{8h}.
  \]
  We therefore have
  \[
    \max_{v\in V}\frac{\beta_v}{\min_{i\in\calL_v}\delta_{v,i}+1} \leq 8h.
  \]
  Because we have $h'\geq \log(8h)$, the choice $\bar{\tau}=\tau(h', [h], m)$ satisfies the requirements of \Cref{lemma:SimpleOLDC} and we can therefore compute the $\gamma$-classes $i_v$ for all nodes $v$ by using the algorithm of \Cref{lemma:SimpleOLDC}.

  We will now analyze the required message size and round complexity of the algorithm. In the first phase the OLDC problem on color lists $\mathcal{L}_v$ and defects $\delta_{v, i_v}$ has to be solved. As $|\mathcal{L}_v| \leq h$ and  transmitting such a single defect does not need more than $\log h$ bits, by  \Cref{lemma:SimpleOLDC}, solving this OLDC instance the maximum message size is  $O(h + \log h+ \log m) = O(h + \log m)$ bits. The number of rounds needed for this first phase are $O(h')$. In the second phase we have messages of size $O(\min \{ |\calC| + \Lambda \log |\calC| \} + \log \log \beta)$ (the initial color is already known in the second phase) due to \Cref{lemma:SimpleOLDC}. The round complexity of the second phase is $O(h)$ due to \Cref{lemma:technicalMain}. Combining both phases and using that $h' = O(h) = O(\log \beta)$, the maximum message size and the runtime are as stated. 
\end{proof}


\section{Recursive Color Space Reduction}
\label{sec:spacereduction}

Distributed list defective coloring algorithms first implicitly appeared in \cite{Kuhn20} as a tool to recursively reduce the color space of a distributed coloring problem. In this section, we show that the idea of using list defective colorings to recursively reduce the color space can also be applied directly to the (oriented) list defective coloring problem. In this way, at the cost of requiring slightly larger lists, we can turn a given distributed (oriented) list defective coloring algorithm into another distributed (oriented) list defective coloring algorithm that is faster and/or needs smaller messages. The high-level idea is the following. Assume that we are given an (oriented) list defective coloring problem with colors from a color space $\calC$. We can arbitrarily partition $\calC=\calC_1\cup\dots\cup\calC_p$ into $p$ approximately equal parts. Instead of directly choosing a color, each node $v$ now first just selects the color subspace $\calC_i$ from which $v$ chooses its color. If $v$ starts with color list $L_v$, then after choosing the color subspace $\calC_i$, $v$'s color list reduces to $L_{v,i}=L_v\cap \calC_i$ (with the original defects on those colors). However, $v$ now only has to compete with neighbors that also pick the same color subspace $\calC_i$. The choices of color subspaces by the nodes can itself be phrased as an (oriented) list defective coloring instance for a color space of size $p$ and thus also with lists of size at most $p$. \Cref{thm:spacereduction} in \Cref{sec:contributions}  formalizes this idea. The theorem is stated in a parameterized way for more general oriented list defective coloring algorithms than the one we have given by \Cref{thm:def_local_list_coloring}. 

\restatesecond*
\begin{proof}
  Assume that $\calC$ is the color space for the oriented list defective coloring problem for which we need to build algorithm $\calA'$ and let $k=\lceil\log_p |\calC| \rceil$. W.l.o.g., we can assume that $|\calC|=p^k$ as otherwise, we can just add some dummy colors to $\calC$ without changing the claimed results. We prove the statement by induction on $k$. For $k=1$ (i.e., for the base of the induction), we have $p=|\calC|$, and we can directly set $\calA'=\calA$ to obtain the properties claimed about algorithm $\calA'$. Let us therefore do the induction step from $k-1$ to $k$ and assume that the claimed results hold if $|\calC|=p^{k-1}$.

  We arbitrarily partition the color space $\calC$ of size $|\calC|=p^k$ into $p$ parts $\calC=\calC_1\cup\dots\cup\calC_p$ of size $|\calC_i|=|\calC|/p=p^{k-1}$. Next, each node $v\in V$ has to choose a color subspace $\calC_{i}$. For each node $v\in V$ and every $i\in\set{1,\dots,p}$, we define $L_{v,i}:=L_v\cap \calC_i$ to be the remaining color list if $v$ decides to choose a color from $\calC_i$. For each node $v\in V$, in the following, we use $i_v\in \set{1,\dots,p}$ to denote the subspace $\calC_{i_v}$ that node $v$ picks. Note that after deciding to choose a color from $\calC_{i_v}$, node $v$ is already guaranteed to pick a different color than any outneighbor $w$ with $i_w\neq i_v$. After the nodes pick their subspaces, we therefore have individual independent oriented list defective coloring problems for each color subspace $\calC_i$.

  For each $v\in V$ and each $i\in\set{1,\dots, p}$, we define
  \[
    \lambda_{v,i} := \frac{\sum_{x\in L_{v,i}}(d_v(x)+1)^{1+\nu}}{\beta_v^{1+\nu}\cdot \kappa(p)^k}
    \quad\text{and}\quad
    \beta_{v,i} := \left\lfloor \left(\lambda_{v,i}\cdot \beta_v^{1+\nu}\cdot \kappa(p) \right)^{1/(1+\nu)}\right\rfloor.
  \]
  Note that we therefore have
  \[
    \sum_{x\in L_{v,i}}(d_v(x)+1)^{1+\nu} = \lambda_{v,i}\cdot \beta_v^{1+\nu}\cdot \kappa(p)^k \geq
    \beta_{v,i} ^{1+\nu}\cdot \kappa(p)^{k-1}.
  \]
  If each node $v$ has at most $\beta_{v,i_v}$ outneighbors that pick the same color subspace $\calC_{i_v}$, the induction hypothesis therefore allows us to compute a solution to the remaining oriented list defective problem in $(k-1)\cdot T(p)$ rounds and with messages of at most $M(p)$ bits. It therefore remains to show that we can assign color subspaces $i_v\in \set{1,\dots,p}$ such that each node $v\in V$ has at most $\beta_{v,i_v}$ outneighbors that also choose color subspace $\calC_{i_v}$. For the problem of choosing the color subspace $i_v$ of each node, we use algorithm $\calA$. The problem can be phrased as an oriented list defective coloring problem, where the color list of every node is $\set{1,\dots,p}$ and the defect that node $v$ is allowed to have for color $i$ is $\beta_{v,i}$. Note that for all $v\in V$, we have
  \[
   \sum_{i=1}^p (\beta_{v,i}+1)^{1+\nu} \geq \sum_{i=1}^p \lambda_{v,i}\cdot \beta_v^{1+\nu}\cdot \kappa(p) \geq \beta_v^{1+\nu}\cdot \kappa(p).
  \]
  The last inequality follows because for every $v\in V$, we have $\sum_{i=1}^p \lambda_{v,i} \geq 1$. This follows by the definition of $\lambda_{v,i}$ and because
  \[
    \sum_{i=1}^p \sum_{x\in L_{v,i}} (d_v(x)+1)^{1+\nu} = \sum_{x\in L_{v}} (d_v(x)+1)^{1+\nu} \geq \beta_v^{1+\nu}\cdot \kappa(p)^k.
  \]
  The oriented list defective problem that we need to solve for choosing the color subspaces $i_v$ therefore satisfies the requirements for algorithm $\calA$, and we can therefore choose the subspaces of all nodes in time $T(p)$ and with messages of size $M(p)$. This concludes the induction step and the proof.
\end{proof}

We remark that when replacing $\beta$ by $\Delta$ and $\beta_v$ by $\deg(v)$, the result of \Cref{thm:spacereduction} also holds for the list defective coloring problem in undirected graphs. To see this, assume that we are given an undirected graph $G$. By replacing every undirected edge $\set{u,v}$ by two directed edges $(u,v)$ and $(v,u)$, an oriented list defective coloring problem on the resulting directed graph is equivalent to the corresponding list defective coloring problem on $G$.

 It has been well-known since Linial's seminal work in \cite{Linial1987} that in directed graphs of outdegree at most $\beta$, one can compute a proper $O(\beta^2)$-coloring in $O(\log^* n)$ rounds (or in $O(\log^* m)$ rounds if an initial proper $m$-coloring is provided). In \cite{Kuhn2009}, it was shown that in the same way, one can also compute an oriented $d$-defective coloring with $O((\beta/d)^2)$ colors. In \cite{MausT20}, the coloring result of \cite{Linial1987} was extended to the list coloring problem and in \Cref{sec:OLDC} of this paper (and to a limited extent also in \cite{MausT20}), the defective coloring result of \cite{Kuhn2009} is extended to the oriented list defective coloring problem. While there has been progress on solving the natural list and defective coloring variants of $O(\beta^2)$ coloring, it is still unknown if a coloring with $O(\beta^{2-\eps})$ colors (for some constant $\eps>0$) can be computed in time $f(\beta) + O(\log^* n)$.\footnote{Note that oriented graphs with maximum outdegree $\beta$ have colorings with $O(\beta)$ colors. However, the best distributed algorithm to compute an $O(\beta)$-coloring requires time $O(\log^3\beta \cdot \log n)$~\cite{GhaffariKuhn21}. It is not known if colorings with $O(\beta^{2-\eps})$ colors can be computed with an $n$-dependency $o(\log n)$.} Even if a moderately fast distributed algorithm for better oriented list defective colorings exists, we directly also get much faster algorithms for computing proper colorings with $o(\beta^2)$ colors. In the following, we assume that there exists an oriented defective coloring algorithm with a round complexity that is polynomial in the number of colors per node plus $O(\log^* n)$. Such algorithm, for example, exist for (list) defective colorings in graphs of neighborhood independence at most $\Delta^\eps$~\cite{barenboim11,Kuhn20}.

 \begin{corollary}\label{cor:spacereduction_time}
  Let $\nu \geq 0$ be a parameter and let $\kappa(\Lambda)$ be a non-decreasing functions of the maximum list size $\Lambda$. Assume that we are given a deterministic distributed algorithm $\calA$ that solves oriented list defective coloring instances for which
  \[
    \forall v\in V\,:\,\sum_{x\in L_v}\big(d_v(x)+1\big)^{1+\nu} \geq \beta_v^{1+\nu}\cdot \kappa(\Lambda).
  \]
  Assume further that if an initial proper $m$-coloring is given, $\calA$ has a round complexity of $\poly(\Lambda) + O(\log^* m)$. Then, there exists a $\big(2^{O(\sqrt{\log\beta \log\kappa(\Lambda)})} + O(\log^* m)\big)$-round deterministic distributed list coloring algorithm $\calA'$ to solve list coloring instances with colors from a color space of size $\poly(\beta)$ for which
  \[
    \forall v\in V\,:\,\sum_{x\in L_v} (d_v(x)+1)^{1+\nu} \geq \beta_v^{1+\nu}\cdot 2^{O(\sqrt{\log\beta \cdot \log\kappa(\Lambda)})}.
  \]
\end{corollary}
\begin{proof}
  At the beginning $\calA'$ uses a standard $O(\log^* m)$-round algorithm of \cite{Linial1987} to properly color $G$ with $O(\beta^2)$ colors. As a result, the time complexity of algorithm $\calA$ afterwards becomes $\poly(\Lambda) + O(\log^* \beta)$. We choose $p=2^{\Theta(\sqrt{\log\beta \log\kappa(\Lambda)})}$ such that $|\calC|=\poly(\beta)=p^{\Theta(\sqrt{\log \beta / \log \kappa(\Lambda)})}$ and we therefore have $\log_p |\calC|= \Theta(\sqrt{\log\beta / \log \kappa(\Lambda)})$. By \Cref{thm:spacereduction}, we then have that algorithm $\calA'$ has a round complexity of $\poly(p)\cdot\log_p |\calC| = 2^{O(\sqrt{\log \beta \log \kappa(\Lambda)})}$ and it requires lists for which
    \[
      \sum_{x\in L_v} (d_v(x)+1)^{1+\nu} \geq \beta_v^{1+\nu}\cdot \kappa(\Lambda)^{\log_p |\calC|} = \beta_v^{1+\nu}  \cdot 2^{O(\sqrt{\log\beta \log\kappa(\Lambda)})}.
    \]
    Where the last equality follows from $\log \kappa(\Lambda) \cdot {\log_p |\calC|} = O \left(\sqrt{\log\beta \log\kappa(\Lambda)}\right)$. This concludes the proof of the corollary.
  \end{proof}
  
Note that the $2^{O(\sqrt{\log \Delta})} + O(\log^* n)$-round algorithm for computing a $(\Delta+1)$-coloring in graphs of bounded neighborhood independence and thus in particular in line graphs of bounded rank hypergraphs is based on the same idea as \Cref{cor:spacereduction_time}. The corollary shows how in some cases, recursive color space reduction can be used to significantly speed up a given (oriented) list defective coloring problem. The following corollary shows that recursive color space reduction can sometimes also be used to significantly reduce the required message size of an (oriented) list defective coloring algorithm. In the oriented list defective coloring algorithm of \Cref{sec:OLDC}, all nodes need to learn the lists and defect vectors of their neighbors and this dominates the required communication. A list $L_v$ of length $|L_v|\leq \Lambda$ consisting of colors from a color space of size $|\calC|$ can be represented by $\min\set{|\calC|, \Lambda\log |\calC|}$ bits and a corresponding defect vector can be represented by $\Lambda\log\beta$ bits, or even by $\Lambda\log\log \beta$ bits if we assume that all defects are integer powers of $2$ (which can usually be assumed at the cost of a factor $2$ in the required list size). In the following, we assume that we have an algorithm that requires $O(|\calC|\cdot B + \log n)$ bits for some parameter $B\geq 1$ (the $\log n$ is included to cover things like exchanging initial colors, unique IDs, etc.).

\begin{corollary}\label{cor:spacereduction_msg}
  Let $\nu \geq 0$ be a parameter and let $\kappa(\Lambda)$ be a non-decreasing functions of the maximum list size $\Lambda$. Assume that we are given a deterministic distributed algorithm $\calA$ that solves oriented list defective coloring instances for which
  \[
    \forall v\in V\,:\,\sum_{x\in L_v}\big(d_v(x)+1\big)^{1+\nu} \geq \beta_v^{1+\nu}\cdot \kappa(\Lambda).
  \] 
  Assume further that $\calA$ has a round complexity of $T(\Lambda)$ and requires messages of $O(|\calC|\cdot B + \log n)$ bits, where $\calC$ the color space and $B\geq 1$ is some parameter. Then, for every integer $r\geq 1$, there exists an $O(T(\Lambda)\cdot r)$-round deterministic distributed list coloring algorithm $\calA'$ to solve list coloring instances with colors from the same color space and for which
  \[
    \forall v\in V\,:\,\sum_{x\in L_v} (d_v(x)+1)^{1+\nu} = \beta_v^{1+\nu}\cdot \kappa(\Lambda)^r.
  \]
  The algorithm $\calA'$ requires messages of size $O(|\calC|^{1/r}\cdot B + \log n)$.
\end{corollary}
\begin{proof}
  We choose $p=\lceil |\calC|^{1/r}\rceil$ so that $p^r\geq |\calC|$ and thus the color space can be recursively partitioned in $r$ steps. The corollary then follows directly from \Cref{thm:spacereduction}.
\end{proof}

As for \Cref{thm:spacereduction}, when replacing $\beta_v$ by $\deg(v)$, \Cref{cor:spacereduction_time} and \Cref{cor:spacereduction_msg} both also hold for the list defective coloring problem in undirected graphs.


\section{Applying List Defective Colorings}
\label{sec:arbdefective}

In \cite{barenboim16sublinear} and \cite{fraigniaud16local}, Barenboim, and Fraigniaud, Heinrich, and Kosowski developed a technique to transform fast, but relaxed (oriented) list coloring into efficient algorithms for the $(\mathit{degree}+1)$-list coloring problem. The same technique has later also been used by the algorithms in \cite{Kuhn20,BalliuKO20,BalliuBKO22}. The high-level idea of this transformation is as follows. Assume that for some $\alpha>1$, we have a $T$-round algorithm $\calA$ that solves list coloring instances with lists of size $>\alpha\Delta$ in graphs of maximum degree $\Delta$. We can first use a defective $k$-coloring to decompose the graph into $k$ subgraphs of maximum degree $\leq \Delta/(2\alpha)$. One then iterates over those color classes and extends a given partial $(\mathit{degree}+1)$-list coloring. When working on the nodes of some color class, all nodes that still have at least $\Delta/2$ uncolored neighbors also still have a list of size $>\Delta/2$. This is more than $\alpha$ times the maximum degree $\Delta/(2\alpha)$ in the current color class, and we can therefore color such nodes by using algorithm $\calA$. In $k\cdot T$ rounds, we can therefore reduce the maximum degree of our $(\mathit{degree}+1)$-list coloring problem from $\Delta$ to $\Delta/2$ and by repeating $O(\log\Delta)$ times, we can solve the $(\mathit{degree}+1)$-list coloring problem. If the algorithm $\calA$ works on directed graphs of maximum outdegree $\beta$ and requires lists of size $>\alpha\beta$, the same idea also works if we decompose the graph by using an arbdefective coloring instead of a defective coloring.

The contribution of this section is two-fold. First, we show that if we assume the existence of (oriented) list coloring algorithms that are significantly better than the current state of the art, we would directly obtain significantly faster algorithms for the standard $(\Delta+1)$-coloring problem. Moreover, we show that by replacing the algorithm $\calA$ in the description above by an (oriented) list defective coloring algorithm, the technique cannot only be used for the $(\mathit{degree}+1)$-list coloring problem, but it also works for computing arbdefective colorings and more generally list arbdefective colorings. In fact, it works for list arbdefective colorings with lists $L_v$ and defects $d_v$ such that for all $v\in V$, $\sum_{x\in L_v}(d_v(x)+1)>\deg(v)$. In the following, we refer to such instances as $(\mathit{degree}+1)$-list arbdefective coloring instances. We subsequently assume that $\calA_{\nu,\kappa}^D$ is a deterministic distributed list defective coloring algorithm that operates on undirected graphs and $\calA_{\nu,\kappa}^O$ is a deterministic distributed oriented list defective coloring algorithm that operates on directed graphs. For real values $\nu \geq 0$ and $\kappa > 0$ we assume that $\calA_{\nu,\kappa}^D$ and $\calA_{\nu,\kappa}^O$ solve all (oriented) list defective coloring problems for which for all $v\in V$,

\begin{eqnarray}\label{eq:betterdefcoloringcond1}
  \sum_{x\in L_v}\big(d_v(x)+1\big)^{1+\nu} & \geq & \deg(v)^{1+\nu}\cdot \kappa\ \text{ and }\\
  \sum_{x\in L_v}\big(d_v(x)+1\big)^{1+\nu} & \geq & \beta_v^{1+\nu}\cdot \kappa, 
  \label{eq:betterdefcoloringcond2}
\end{eqnarray}
respectively. We assume that the round complexity of algorithm $\calA_{\nu,\kappa}^D$ is $T_{\nu,\kappa}^D$ and that the round complexity of algorithm $\calA_{\nu,\kappa}^O$ is $T_{\nu,\kappa}^O$. \Cref{thm:arbdefective} in \Cref{sec:contributions} shows that by using $\calA_{\nu,\kappa}^O$, one can solve $(\mathit{degree}+1)$-list arbdefective coloring instances in time $O\big(\Lambda^{\frac{\nu}{1+\nu}}\cdot \kappa^{\frac{1}{1+\nu}}\cdot \log(\Delta)\cdot T_{\nu,\kappa}^O + \log^* n\big)$ and by using  $\calA_{\nu,\kappa}^D$ one can solve such list arbdefective colorings in time $O\big(\Lambda^{\nu}\cdot \kappa^2\cdot \log(\Delta)\cdot T_{\nu,\kappa}^D + \log^* n\big)$. 
\restatethird*

\begin{proof} 
  We show how to reduce the original list arbdefective coloring problem on $G$ to a list arbdefective coloring problem on a subgraph of $G$ of maximum degree less than $\Delta/2$. We first argue about how to use the oriented list defective coloring algorithm $\calA_{\nu,\kappa}^O$ to achieve this and we then discuss what changes when using the list defective coloring algorithm $\calA_{\nu,\kappa}^D$. As a first step, we compute an arbdefective coloring of $G$ with
  \begin{equation}
    \label{eq:arbcoloring}
    q := O\left(\Lambda^{\frac{\nu}{1+\nu}}\cdot \kappa^{\frac{1}{1+\nu}}\right) \mathit{colors}\quad\text{and}\quad
    \mathit{arbdefect}\ \delta := \frac{\Delta}{2}\cdot \frac{1}{\Lambda^{\frac{\nu}{1+\nu}}\cdot\kappa^{\frac{1}{1+\nu}}}
  \end{equation}
  Note that $q\cdot \delta = \Theta(\Delta)$ and if at the beginning, we spend $O(\log^* n)$ rounds to compute a proper $O(\Delta^2)$-coloring of $G$, such an arbdefective coloring can therefore be computed in $O(q)$ rounds by using an algorithm of \cite{BarenboimEG18}. For each $i\in \set{1,\dots,q}$, let $V_i$ be the set of nodes that are colored with color $i$ in this arbdefective coloring and let $G_i$ be the directed graph that is given by the subgraph of $G$ induced by the nodes in $V_i$ together with the outdegree $\leq \delta$ edge orientation of this graph that is given as part of the arbdefective coloring. We then iterate through the colors $i\in \set{1,\dots,q}$. In iteration $i$, the goal is to color a subset of the nodes in $V_i$. Throughout the algorithm, we maintain the following. Each node is either colored or uncolored, and a node that is colored remains with this color for the rest of the algorithm. As soon as two neighbors in $G$ are both colored, the edge between them is oriented and remains oriented in this way until the end of the algorithm. For each node $v\in V$ and each color $x\in L_v$, at all times, we use $a_v(x)$ to denote the number of colored neighbors that are colored with color $x$. We make sure that for every colored node $v\in V$, if $v$ is colored with color $x\in L_v$, then the number of neighbors that are colored with color $x$ and for which the edge is pointing away from $v$ is at most $d_v(x)$.

  Let us now focus on the phase, where we process the nodes in $V_i$. Note that at this point, a subset of the nodes in $V_1\cup\dots\cup V_{i-1}$ are already colored and the remaining nodes are still uncolored. Let $V_i'\subseteq V_i$ be the subset of the nodes in $V_i$ that still have at least $\Delta/2$ uncolored neighbors in $G$. Our goal is to color the nodes in $V_i'$ by using the oriented list defective algorithm $\calA_{\nu,\kappa}^O$ on the directed graph $G_i':=G_i[V_i']$. Note that the maximum outdegree in $G_i'$ is at most $\delta$. For each node $v\in V_i'$, we define a list $L_v'$ and a defect function $d_v'$ as follows. The list $L_v'$ contains all colors $x\in L_v$ for which $a_v(x)\leq d_v(x)$. For each color $x\in L_v'$, we then set $d_v'(x) = d_v(x) - a_v(x)$. We then have
  \[
    \forall v\in V_i'\,:\, \sum_{x\in L_v'}(d_v'(x)+1) \geq \sum_{x\in L_v}(d_v(x)+1) - \sum_{x\in L_v} a_v(x) > \deg(v) - (\deg(v) - \lceil\Delta/2\rceil) \geq \Delta/2.
  \]
  For the second inequality, note that the sum over the $a_v(x)$ is equal to the number of colored neighbors of $v$, which is at most $\deg(v) - \lceil\Delta/2\rceil$ because $v$ still has at least $\lceil\Delta/2\rceil$ uncolored neighbors. For all $v\in V_i'$, we then have
  \begin{eqnarray*}
    \sum_{x\in L_v'} (d_v'(x) + 1)^{1+\nu}
    & \geq & \frac{\left(\sum_{x\in L_v'}(d_v'(x) +1)\right)^{1+\nu}}{|L_v'|^{\nu}}\ >\ \frac{(\Delta/2)^{1+\nu}}{|L_v'|^{\nu}}\\
    & \geq & \frac{(\Delta/2)^{1+\nu}}{\Lambda^\nu\cdot \kappa}\cdot \kappa \ =\ \delta^{1+\nu}\cdot\kappa. 
  \end{eqnarray*}
  The first inequality follows from H\"oder's inequality, which is a generalization of the Cauchy-Schwarz inequality to general $l_p$ norms.
  We can therefore use algorithm $\calA_{\nu,\kappa}^O$ to compute an oriented list defective coloring of $G_i'$ (and thus of the nodes in $V_i'$). The colors of $V_i'$ extend the partial list arbdefective coloring of $G$ as follows. All the edges from $V_i'$ to colored nodes in $V_1\cup\dots\cup V_{i-1}$ are oriented away from the node in $V_i'$ and the orientations of the edges between two nodes in $V_i'$ are given by the directions of those edges in $G_i'$. The nodes in $V_1\cup\dots\cup V_{i-1}$ then do not get any new outgoing edges and thus their arbdefect condition trivially still holds. If a node $v\in V_i'$ is colored with color $x\in L_v'$, it obtains at most $d_v'(x)$ additional outneighbors of color $x$ and its total number of outneighbors of color $x$ is therefore $\leq a_v(x) + d_v'(x) = d_v(x)$. After iterating over all $q$ color classes of the arbdefective coloring of $G$, we therefore have a partial list arbdefective coloring of $G$ such that a) every uncolored node $v\in V$ has at most $\Delta/2$ uncolored neighbors (otherwise, $v$ would have been added to the respective set $V_i'$), b) all edges between colored nodes are oriented, and c) for every color $x$, every colored node $v$ of color $x$ has at most $d_v(x)$ colored outneighbors of color $x$. The time for computing this partial list arbdefective coloring is 
  \begin{equation*}
    O\big(q\cdot T_{\nu,\kappa}^O\big) = O\left(\Lambda^{\frac{\nu}{1+\nu}}\cdot\kappa^{\frac{1}{1+\nu}}\cdot T_{\nu,\kappa}^O\right).
  \end{equation*}
  At the end, we orient all edges between uncolored and colored nodes from the uncolored to the colored node. In this way, the nodes that are already colored will remain valid independent of how we color the rest of the graph. Let $\bar{G}=(\bar{V}, \bar{E})$ be the subgraph of $G$ induced by the uncolored nodes. For each node  $v\in \bar{V}$ and each color $x\in L_v$, $v$ already has $a_v(x)$ neighbors of color $x$. Hence, when coloring the graph $\bar{G}$, $v$ can only tolerate $d_v'(x) = d_v(x) - a_v(x)$ additional neighbors of color $x$. We therefore obtain a new list arbdefective coloring problem on $\bar{G}$, where each node $v\in \bar{V}$ gets a list $L_v'$ that contains all colors $x\in L_v$ for which $a_v(x)\leq d_v(x)$ and the allowed arbdefects for $x\in L_v'$ are $d_v'(x) = d_v(x) - a_v(x)$. Note that since $\sum_{x\in L_v}(d_v(x)+1)>\deg_G(v)$, and $\sum_{x\in L_v}a_v(x) = \deg_G(v)-\deg_{\bar{G}}(v)$ (i.e., the number of colored neighbors of $v$), we have $\sum_{x\in L_v'}(d_v'(x)+1)>\deg_{\bar{G}}(v)$. Note that we clearly also have $\sum_{x\in L_v'}(d_v'(x)+1)\geq |L_v'|$ and we can therefore reduce the number of colors in $L_v'$ to at most $\deg_{\bar{G}}(v)+1$. We can therefore proceed with the graph $\bar{G}$ in the same way and again compute a partial coloring to reduce the maximum degree of the uncolored part of the graph from $\Delta/2$ to $\Delta/4$. Let us define the phase where we compute a partial coloring of the current graph to reduce the maximum degree from $\leq \Delta/2^{i-1}$ to $\leq \Delta/2^i$ as stage $i > 0$. Also, let $\Lambda_i$ be the maximum list length in stage $i$. Note that $\Lambda_i\leq \min\set{\Lambda, \Delta/2^{i-1}+1}$. For the overall time, we then obtain
  \[
    O\left(\sum_{i=1}^{\log\Delta} \Lambda_i^{\frac{\nu}{1+\nu}}\cdot\kappa^{\frac{1}{1+\nu}}\cdot  T_{\nu,\kappa}^O\right).
  \]
  This can be upper bounded by $O\big(\Lambda^{\frac{\nu}{1+\nu}}\cdot \kappa^{\frac{1}{1+\nu}}\cdot \log(\Delta)\cdot T_{\nu,\kappa}^O\big)$ and by $O\big(\Lambda^{\frac{\nu}{1+\nu}}\cdot \kappa^{\frac{1}{1+\nu}}\cdot \log(\frac{\Delta}{\Lambda})\cdot T_{\nu,\kappa}^O\big)$ if $\nu\geq \nu_0$ for some constant $\nu_0>0$. 

  Let us now also discuss the differences when using a list defective algorithm $\calA_{\nu,\kappa}^D$ instead of an oriented list defective algorithm $\calA_{\nu,\kappa}^O$. In this case, we decompose the graph by using a defective coloring instead of an arbdefective coloring. Unfortunately, there are no defective coloring algorithms that are similarly efficient as the arbdefective coloring algorithm of \cite{BarenboimEG18}. Instead, we can compute the defective coloring of $G$ by using our algorithm $\calA_{\nu,\kappa}^D$. We can use $\calA_{\nu,\kappa}^D$ to compute a defective coloring with $q$ colors and defect $\delta$ as long as $q\cdot (\delta+1)^{1+\nu} > \Delta^{1+\nu}\cdot\kappa$ (due to \ref{eq:betterdefcoloringcond1}). For the list arbdefective coloring algorithm to work, we need to use the same defect $\delta$ as before. However, the number of colors $q$ now becomes larger. We now get $q=\Theta\big(\big(\frac{\Delta}{\delta}\big)^{1+\nu}\cdot\kappa\big) = \Theta\big(\Lambda^{\nu}\cdot\kappa^2\big)$. The message complexity follows because apart from running the algorithm $\calA_{\nu,\kappa}^D$ as a subroutine, we only need to use the algorithms of \cite{Linial1987} and \cite{BarenboimEG18} to compute the initial proper $O(\beta^2)$-coloring and the arbdefective $q$-coloring. Both algorithms operate in the \CONGEST model and require only $O(\log n)$-bit messages.

  The rest of the proof now follows in almost the same way as before. The graph $G_i$ for which we compute a partial list arbdefective coloring is then an undirected graph. For the neighbors that are not colored in the same instance of $\calA_{\nu,\kappa}^D$, we still orient the edge from the node that is colored later to the node that is colored first. Edges between nodes that are colored by the same instance of $\calA_{\nu,\kappa}^D$ can be colored arbitrarily. 
\end{proof}

\subsection{Implications of \Cref{thm:arbdefective}}

We first discuss two immediate implications of \Cref{thm:arbdefective}, and we afterwards show how the theorem can be used to improve the best current deterministic complexity of the $(\Delta+1)$-coloring problem in the \CONGEST model.

\para{Complexity of Computing (List) Arbdefective Colorings.} For a first immediate implication of \Cref{thm:arbdefective}, we can use the algorithm of \Cref{thm:def_local_list_coloring} as the oriented list defective coloring algorithm $\calA_{\nu,\kappa}^O$. If we assume that the color space that we have is of size $|\calC|=\poly(\beta)$, in this case, $\nu=1$ and $\kappa=O(\log\beta \cdot \log^3\log\beta)$. This results in an arbdefective coloring algorithm that solves instances with lists $L_v$ for which $\forall v\in V\,:\,\sum_{x\in L_v}(d_v(x)+1)>\deg(v)$ in time $O\big(\sqrt{\Lambda}\cdot \log^{5/2}\Delta \cdot \log^{3/2}\log\Delta + \log^* n\big)$. In particular, this implies that for any $d\geq 0$ and any $q>\frac{\Delta}{d+1}$, a $d$-arbdefective $q$-coloring can be computed in time $O\big(\sqrt{\frac{\Delta}{d+1}}\cdot\log^{5/2}\Delta \cdot\log^{3/2}\log\Delta + \log^* n\big)$, which significantly improves the previously best algorithms that achieves the same arbdefective coloring in time $O(\Delta+\log^* n)$~\cite{balliu2021hideandseek} or a more relaxed $d$-arbdefective $O\big(\frac{\Delta}{d+1}\big)$-coloring in time $O\big(\frac{\Delta}{d+1}+\log^* n\big)$. Note also that the condition $\forall v\in V\,:\,\sum_{x\in L_v}(d_v(x)+1)>\deg(v)$ is necessary in order to compute a (list) arbdefective coloring in time $f(\Delta)+O(\log^* n)$. If the condition does not hold, any deterministic algorithm for the problem requires at least $\Omega(\log_\Delta n)$ rounds~\cite{balliu2021hideandseek}.

\para{Better List Defective Coloring Implies Better \boldmath$(\Delta+1)$-Coloring.} The theorem in particular also implies that certain progress on (oriented) list defective coloring algorithms would directly lead to faster algorithms for the standard $(\Delta+1)$-coloring problem. Assume that for an initial $m$-coloring of the graph, we have an oriented list defective coloring algorithm with a round complexity that is $\poly(\Lambda) + O(\log^* m)$ and that satisfies equation \eqref{eq:betterdefcoloringcond1} for any constant $\nu<1$. In combination with \Cref{cor:spacereduction_time}, \Cref{thm:arbdefective} then implies that we then obtain a $(\mathit{degree}+1)$-list coloring (and thus $(\Delta+1)$-coloring) algorithm with a time complexity of $O\big(\Delta^{\frac{\nu}{1+\nu}+o(1)} + \log^* n\big)$, which would be polynomial improvement over the $O(\sqrt{\Delta\log\Delta}+\log^* n)$-round algorithm of \cite{fraigniaud16local,BarenboimEG18,MausT20}. The same would be true if we had a list defective coloring algorithm with a round complexity of $\poly(\Lambda) + O(\log^* m)$ and that satisfies equation \eqref{eq:betterdefcoloringcond2} for any constant $\nu<1/2$. We believe that if it is possible to significantly improve the current best $O(\sqrt{\Delta\log\Delta}+\log^* n)$-round complexity of $(\Delta+1)$-coloring, the key will be to better understand the distributed complexity of (oriented) defective colorings and probably also of the more general (oriented) list defective colorings.

\para{Complexity of \boldmath$(\Delta+1)$-Coloring in the \CONGEST Model.} Apart from the standard $(\Delta+1)$-coloring problem, in the following, we also consider the general $(\mathit{degree}+1)$-list coloring problem. In order to keep the results simple and because this is the most interesting case, we will assume that we have $(\mathit{degree}+1)$-list coloring instances with a color space of size at most $\poly(\Delta)$. Note that in the case of the standard $(\Delta+1)$-coloring problem, the color space is of size $\Delta+1$. For small $\Delta$, the best $(\Delta+1)$-coloring algorithm in the \LOCAL model has a round complexity of $O(\sqrt{\Delta\log\Delta} + \log^* n)$~\cite{fraigniaud16local,BarenboimEG18,MausT20} and the following theorem shows that this round complexity can almost be matched in the \CONGEST model.

\restatecongest*
\begin{proof}
  If $\Delta$ is moderately large, the fastest known deterministic $(\mathit{degree}+1)$-list coloring algorithm in the \CONGEST model (for lists of size $\poly(\Delta)$) has a time complexity of $O(\log^2\Delta \cdot\log n)$ due to \cite{GhaffariKuhn21}. If $\Delta > \log^2 n$, the time complexity of this algorithm is $O(\sqrt{\Delta}\cdot\log^2\Delta)$. We can therefore assume that $\Delta \leq \log^2 n$.

  As a first step, we use a standard algorithm of \cite{Linial1987} to properly color $G$ with $O(\Delta^2)$ colors in $O(\log^* n)$ rounds. We then design an oriented list defective coloring algorithm $\calA$ that can be implemented efficiently in the \CONGEST model and that can afterwards be used in \Cref{thm:arbdefective}. For this, we use \Cref{cor:spacereduction_msg} in combination with the algorithm of \Cref{thm:def_local_list_coloring}. For graphs with outdegree $\beta=\poly\Delta$, the algorithm of \Cref{thm:def_local_list_coloring} requires messages of $O(|\calC| + \log \Delta)$ bits. If we assume that $|\calC|=O(\Delta^p)$ and set the parameter $r$ in \Cref{cor:spacereduction_msg} to $r=2p$, we obtain an algorithm $\calA$ with $O(\sqrt{\Delta} + \log n)=O(\log n)$-bit messages (recall that we assume $\Delta\leq\log^2 n$). For $\beta=\poly\Delta$, the algorithm has a round complexity of $O(\log\Delta)$.  Using the fact that the color space and the given proper coloring of $G$ are both of size $\poly\Delta$, \Cref{thm:def_local_list_coloring} together with \Cref{cor:spacereduction_msg} imply that the resulting algorithm $\calA$ requires lists $L_v$ such that
  \begin{equation}\label{eq:conditionA}
    \forall v\in V\,:\,\sum_{x\in L_v}(d_v(x)+1)^2 \geq \beta_v^2\cdot \Theta\big(\log^{2p}\Delta\cdot \log^{6p}\log\Delta\big). 
  \end{equation}
  We can now use \Cref{thm:arbdefective} to solve $(\mathit{degree}+1)$-list coloring on $G$. Note that $(\mathit{degree}+1)$-list coloring is a special case of list arbdefective coloring, where the size of $v$'s list is $\deg(v)+1$ and all defects are equal to $0$. We therefore have $\Lambda=\Delta+1$. As discussed, algorithm $\calA$ has a round complexity of $T_{\calA} =O(\log\Delta)$ and all lists must require the condition given by \eqref{eq:conditionA}. \Cref{thm:arbdefective} implies that the resulting $(\mathit{degree}+1)$-list coloring algorithm has a round complexity of $O\big(\Lambda^{\frac{\nu}{1+\nu}}\cdot\kappa_{\calA}^{\frac{1}{1+\nu}}\cdot \log\big(\frac{\Delta}{\Lambda}\big)\cdot T_{\calA}\big)$, where $\Lambda=\Delta+1$, $\nu=1$, and $\kappa_{\calA}=\Theta\big(\log^{2p}\Delta\cdot \log^{6p}\log\Delta\big).$ We therefore get the round complexity claimed by the theorem statement.
\end{proof}



\bibliography{references}
\clearpage

\appendix

\section{Sequential List Defective Coloring Algorithms}
\label{sec:existential}

In this section, we give sequential algorithms to compute list defective and list arbdefective colorings. The two algorithms give sufficient conditions on when such colorings exist. We will see that both conditions are also necessary on the complete graph $K_{\Delta+1}$ if all nodes have the same lists.

\begin{lemma}\label{lemma:defectiveexistence}
  Let $G=(V,E)$ be a graph and assume that we are given a list defective coloring instance with lists $L_v\subseteq\calC$ and defect functions $d_v$. The given list defective coloring instance can be solved if
  \[
    \forall v\in V\,:\,\sum_{x\in L_v} (d_v(x)+1) > \deg_G(v).
  \]
\end{lemma}
\begin{proof}
    The proof is a natural generalization of a proof in \cite{lovasz66} for the (standard) special case, where all lists are the same and all colors have the same defect. We start with an arbitrary initial coloring, where each node $v\in V$ is assigned an arbitrary color from its list $L_v$. We then iteratively transform this coloring into one that satisfies all the defect requirements. In the following, we use $x_v\in L_v$ to denote the current color of a node $v\in V$. Define a node $v\in V$ to be unhappy if $v$ is currently colored with a color $x_v$ for which it has more than $d_{v}(x_v)$ neighbors of color $x_v$. As long as there are unhappy nodes, we pick an arbitrary unhappy node $v$ and recolor $v$ with some color $y\in \calC$ for which $v$ has at most $d_v(y)$ neighbors of color $y$. Note that by the assumption that $\sum_{x\in L_v} (d_v(x)+1) > \deg_G(v)$, such a color must always exist for node $v$.
    
    To show that this process converges to a state in which there are no unhappy nodes, we use a potential function argument. For a given coloring, let $M$ be the number of monochromatic edges. We then define a potential $\Phi$ of a given coloring as follows.
    \[
    \Phi := M + \sum_{v\in V}(\deg_G(v) - d_{v}(x_v)).
    \]
    Note that initially, $\Phi \leq 3|E|$ and at all times, we have $\Phi\geq 0$. Now, consider the change of the potential function in one recoloring step. Assume that node $v$ is recolored from color $x_v$ to color $y$. Let $\Phi$ be the potential before the recoloring and let $\Phi'$ be the potential after the recoloring. We have
    \[
    \Phi' - \Phi \leq d_v(y) - (d_v(x_v)+1) + (\deg_G(v) - d_v(y)) - (\deg_G(v) - d_{v}(x_v)) = -1. 
    \]
    The potential therefore strictly decreases in each step until we reach a state where there are no unhappy nodes. Note that if the potential reaches $0$, we have $0$ monochromatic edges and $d_{v}(x_v)=\deg_G(v)\geq 0$ for all $v$ and, in this case, we therefore also clearly do not have any unhappy nodes.
\end{proof}

Note that the lemma is tight for some graphs if all nodes have the same list. If $G=K_{\Delta+1}$ is a complete graph of size $\Delta+1$ and all nodes have the same list $L$ and defect function $d$ with $\sum_{x\in L}(d(x)+1)=\Delta$, then by pigeonhole principle, there must be some color $x\in L$ for which at least $d(x)+2$ nodes of $G$ are colored with color $x$. 

\begin{lemma}\label{lemma:arbdefectiveexistence}
    Let $G=(V,E)$ be a graph and assume that we are given a list arbdefective coloring instance with lists $L_v\subseteq\calC$ and defect functions $d_v$. The given list arbdefective coloring instance can be solved if
  \[
    \forall v\in V\,:\,\sum_{x\in L_v} (2d_v(x)+1) > \deg_G(v).
  \]
\end{lemma}
\begin{proof}
    We first compute a list defective coloring for the lists and the defect function $d_v'$ with $d_v'(x)=2d_v(x)$ for all $x\in L_v$. By \Cref{lemma:defectiveexistence}, such a coloring exists. For each color $x\in \calC$, we now define $G_x=(V_x,E_x)$ to be the subgraph of $G$ induced by the nodes $V_x\subseteq V$ of color $x$. For each $v\in V_x$, let $\delta_x(v)$ be the degree of node $v$ in $G_x$. Note that $\delta_x(v)\leq 2d_v(x)$. In $G_x$, we can orient the edges such that each node $v\in V_x$ has outdegree at most $d_v(x)$ in $G_x$. To see this, consider a graph $G_x'$ that is obtained from $G_x$ by adding an arbitrary perfect matching on the set of odd-degree vertices of $G_x$. The graph $G_x'$ has only even degree and it thus has an Euler tour. If we orient all edges along this Euler tour, then the outdegree of each node $v\in V_x$ is at most $\lceil \delta_x/2\rceil\leq d_v(x)$. Doing this for all colors directly gives a valid solution for the given arbdefective list coloring instance.
\end{proof}

The lemma is again tight if $G=K_{\Delta+1}$ and if all nodes have the same list $L$ and defect function $d$. Assume that $\sum_{x\in L}(2d(x)+1)=\Delta$. Then, by pigeonhole principle, there must be a color $x$ such that $x\in L$ and such that $2d(x)+2$ of the $\Delta+1$ nodes of $G$ are colored with color $x$. On average, those nodes have $d(x)+1$ outneighbors with color $x$ and therefore, there clearly must be one node with color $x$ and with at least $d(x)+1$ outneighbors with color $x$.



\section{Detailed Proof of Lemma \ref{lemma:InterfaceToYannic}}
\label{sec:adaptionOf}
\begin{lemma}[adapted from \cite{MausT20}]
    Let $\gamma,\tau,\tau'\geq 1$ be three integer parameters such that $\tau\geq 8\log\gamma + 2\log\log|\calC| + 2\log\log m + 16$ and $\tau' = 2^{\tau - \lceil\log(2e\gamma^2)\rceil}$. For every color list $L\in \binom{\calC}{\ell}$ of size $\ell$ for some $\ell\geq 2e\gamma^2\tau$, we further define $S(L) := \stackedbinom{L}{\gamma\tau}{\gamma\tau'}$.
      Then, there exists $\bar{S}(L)\subseteq S(L)$ such that $|\bar{S}(L)|\geq |S(L)|/2$ and such that for every $K\in \bar{S}(L)$ and every $L'\in \binom{\calC}{\ell}$, there are at most $d_2<\frac{1}{4m|\calC|^{\ell}}\cdot|S(L)|$ different $K'\in S(L')$ such that $(K, K') \in \Psi(\tau',\tau)$ or $(K', K) \in \Psi(\tau',\tau)$. Further, for every $K\in S(L)$ and every $L'\in \binom{\calC}{\ell}$, there are at most $d_2$ different $K'\in S(L')$ for which $(K,K')\in\Psi(\tau',\tau)$. 
\end{lemma}
\begin{proof}
    Let $k := \gamma \tau$, $L \in \binom{\mathcal{C}}{\ell}$ and let $C \in \binom{L}{k}$. Each $C' \in \binom{L}{k}$ s.t. $|C \cap C'|<\tau$ can be constructed by first choosing $\tau$ elements from $C$ and then adding $k - \tau$ elements of the $\ell - \tau$ remaining colors. This can be done in at most $d_1 := \binom{k}{\tau} \binom{\ell - \tau}{k - \tau}$ many ways. \par
    Let $k' := \gamma \tau'$ and let $K \in S(L)$ of size $k'$. Let $X \subseteq K$ be the set of lists $C' \in \binom{L}{k}$ that are in conflict with at least one element of $C \in K$ i.e. $|C \cap C'|<\tau$. By the above observation $|X| \leq k' d_1$. Each $K' \in S(L)$ s.t. $(K, K') \in \Psi(\tau', \tau)$ can be constructed by adding $\tau'$ lists from $X$ and adding $k' - \tau'$ arbitrary lists. This can be done in at most
    \begin{align*}
        d_2' = \binom{k'd_1}{\tau'} \cdot \binom{\binom{\ell}{k} - \tau'}{k' - \tau'}
    \end{align*}
    many ways. Let $\bar{S}(L) \subseteq S(L)$ be a set s.t. for each element $K \in \bar{S}(L)$ there are at most $d_2 := 4 d_2'$ many $K' \in S(L)$ with $(K, K') \in \Psi(\tau', \tau)$ or $(K', K) \in \Psi(\tau', \tau)$.
    \begin{claim}
        \label{claim:GoodSet}
        $|\bar{S}(L)| \geq |S(L)|/2$
    \end{claim}
    \begin{proof}\renewcommand{\qedsymbol}{$\blacksquare$}
        Let $H = (V_H, E_H)$ be a directed graph over the vertex set $V_H = S(L)$ where $(K, K') \in E_H$ iff $(K, K') \in \Psi(\tau', \tau)$. The outdegree of each node $K$ is by the above analysis, at most $d_2/4$. Hence, the total number of edges is $|E_H| \leq |S(L)| \cdot d_2/4$. Hence, at most half of the nodes do have an (undirected) degree higher than $d_2$. 
    \end{proof}
    We want to show that for for every $K\in \bar{S}(L)$ and every $L'\in \binom{\calC}{\ell}$, there are at most $d_2$ different $K'\in S(L')$ such that $(K, K') \in \Psi(\tau',\tau)$ or $(K', K) \in \Psi(\tau',\tau)$. The above statement shows this for the special case where $L = L'$. Now assume $L \not = L'$. There exists a color $x \in L' \setminus L$. Let $K''$ be the set $K'$ where the color $x$ in each element of $K'$ is replaced with an arbitrary color from $L$. By the definition of $\Psi$, it is clear that if $(K', K) \in \Psi(\tau', \tau)$ then also $(K'', K) \in \Psi(\tau', \tau)$ and if $(K, K') \in \Psi(\tau', \tau)$ then also $(K, K'') \in \Psi(\tau', \tau)$ (note that the other direction does not hold in general). Hence, the number of conflicts is only increasing by replacing colors from $L'$ with colors in $L$.\par
    The statement of the lemma then follows by \cref{claim:BoundDegreed2}. 
\end{proof}
\begin{claim}
    \label{claim:BoundDegreed2}
    $d_2 < \frac{1}{4m|\calC|^{\ell}}\cdot|S(L)|$
\end{claim}
\begin{proof}
\begin{align*}
    \frac{d_2}{|S(L)|} &= 4 \frac{\binom{k' \cdot d_1}{\tau'} \binom{\binom{\ell}{k} - \tau'}{k' - \tau'}}{\binom{\binom{\ell}{k}}{k'}} 
    \\
    &< 4 \frac{\binom{k' \cdot d_1}{\tau'}\binom{\binom{\ell}{k}}{k'} }{\binom{\binom{\ell}{k}}{k'}} \cdot \left( \frac{k'}{\binom{\ell}{k}} \right)^{\tau'} \quad \text{(\cref{claim:approx})} 
    \\
    &\leq 4e^{\tau'} \left(\frac{(k')^2 \binom{k}{\tau}}{\tau'} \right)^{\tau'} \left( \frac{\binom{\ell - \tau}{k - \tau}}{\binom{\ell}{k}} \right)^{\tau'} 
    \\
    &\leq 4e^{\tau'} \left(\frac{(k')^2 e^{\tau} k^{\tau}}{\tau' \cdot \tau^\tau} \right)^{\tau'} \left( \frac{k}{\ell} \right)^{\tau \cdot \tau'} \quad \text{(\cref{claim:approx})} 
    \\
    &\leq 4e^{\tau'} \left(\frac{(\tau') \gamma^2 }{2^\tau} \right)^{\tau'}
    \\
    &\leq 4 \left(\frac{1}{2} \right)^{\tau'}  \\
    &< \frac{1}{4m|\calC|^\ell} \quad \text{(\cref{claim:helper})} 
\end{align*}
In the second last step we used $\tau' \leq \frac{2^\tau}{2e \gamma^2}$.
\end{proof}

\begin{claim}
    \label{claim:approx}
    For integers $L > K > x > 0$ we have 
    \begin{align*}
      \binom{L - x}{K - x} < \left( \frac{K}{L} \right)^x \cdot \binom{L}{K}
    \end{align*} 
\end{claim}
\begin{proof}
    The proof is given in Claim 2 in \cite{MausT20}.
\end{proof}

\begin{claim}
    \label{claim:helper}
    $ 2^{\tau'} > 16 \cdot m |\calC|^\ell$
\end{claim}
\begin{proof}
    By the definition of $\tau$ we have $\tau \geq 16$ (we thus will use $\tau/4 \geq \log \tau$) and $\tau/2 \geq \log \log m + \log \log |\calC| + 4 \log \gamma + 8$.
    
    \begin{align*}
        \log \log (16 \cdot m |\calC|^\ell) &\leq \log \log m + \log \ell + \log \log |\mathcal{C}| + 2 \\
        &\leq \log \log m + \log (2e) + 2 \log \gamma + \log \tau + \log \log |\calC| + 2 \\
        &\leq \log \log m + \log (2e) + 2 \log \gamma + \tau/4  + \log \log |\calC| + 2 \\
        &\leq \tau/4 + \tau/2 - 2 \log \gamma - 6 + \log (2e) \\
        &< \tau - \log (\gamma^2) - \log (4e) \\
        &= \tau - \log (4e \gamma^2)
    \end{align*}
    Hence, $16 \cdot m |\calC|^\ell < 2^{2^{\tau - \log (4e\gamma^2)}} \leq 2^{\tau'}$.
\end{proof}

\section{Complexity of Internal Computations at Nodes}
\label{sec:InternalComplexity}
In the distributed setting (e.g., in \LOCAL and \CONGEST), one typically assumes that internal computations at the nodes are for free and we therefore do not analyze the computational complexity of such internal computations when running a distributed algorithm. In practice, however, it is certainly relevant that all computations are reasonably efficient. We therefore next briefly discuss the cost of internal computations required by our algorithms. In the context of this paper, the most time-intensive (internal) algorithms are used for computing the $0$-round solutions for the problem $P_2$ (in the algorithm of \Cref{thm:def_local_list_coloring}). In the following, we focus on such a computation for a given $\gamma$-class. A simple greedy approach to solve $P_2$ is stated in Appendix $A$ of \cite{MausT20}. Consider some node $v$ and let $S = \{K_1, K_2, ... \}$ be the (ordered) set of all possible lists of color lists that can be constructed from the color space $\mathcal{C}$ and the initial list $L_v$ of node $v$ that is of size $|L_v|=\alpha\gamma^2\tau$. To compute a conflict-free list $S' \subseteq S$ regarding the conflict relation $\Psi(\tau', \tau)$, we add $K_1$ to $S'$ and delete all $K$ from $S$ with $(K_1, K) \in \Psi(\tau', \tau)$ or $(K, K_1) \in \Psi(\tau', \tau)$ and we repeat this process until $S$ is empty. The complexity of this greedy procedure is at most $O(|S|^2)$. We thus have

\begin{align*}
    |S|^2 = \left( \binom{|\mathcal{C}|}{|L_v|} \stackedbinom{|L_v|}{\gamma\tau}{\gamma \cdot \tau'} \right)^2    
     &\leq \left( |\mathcal{C}|^{\alpha \gamma^2 \tau} \cdot \binom{\binom{\alpha \gamma^2 \tau}{\gamma \cdot \tau}}{\gamma \cdot \tau'}\right)^2 \\
    &= e^{O(\gamma^2 \tau \log |\mathcal{C}|)} \cdot (e \alpha \gamma)^{O(\gamma^2 \tau \tau')} \\
    &= e^{O(\gamma^2 \log \gamma \log |\mathcal{C}| + \gamma^3 \log^2 \gamma)}.
\end{align*}

Note that even though there is no dependency on $n$, this complexity is only efficient if $\gamma$ is sufficiently small. Note however that we can use the color space reduction technique of \Cref{thm:spacereduction} to reduce $\gamma$ in each application of \Cref{thm:def_local_list_coloring} and to therefore also reduce the internal computational complexity (at the cost of a somewhat weaker list defective coloring algorithm). To illustrate this, we focus on our $(\Delta+1)$-coloring algorithm for the \CONGEST model (\cref{thm:CONGEST}). When applying the color space reduction of \Cref{thm:spacereduction} with some parameter $p$, the size of the color space and thus also the maximum list size in each application of \Cref{thm:def_local_list_coloring} is at most $p$. If we choose $p = \Delta^{\varepsilon}$ for some $\varepsilon > 0$, we obtain $|\mathcal{C}| = O(\Delta^\varepsilon)$ and $\gamma = O(\Delta^{\epsilon/2})$ (recall that the color lists in \Cref{thm:def_local_list_coloring} are of size $\Omega(\gamma^2)$). Combining this with the fact that $\Delta$ is assumed to be at most $\log^2 n$ (see proof of \cref{thm:CONGEST}) we get $\gamma = O(\log^\varepsilon n)$. The upper bound on $|S|^2$ is thereby simplified to $e^{O(\log^{3 \varepsilon} n \cdot \log^2(\log^\varepsilon n))}$.
Hence, choosing $\varepsilon \leq 1/6$, we get $|S|^2 = o(n)$ and thus a sublinear (in $n$) complexity for all internal computations at the nodes. Consequently, at the cost of some additional $\log\Delta$ factors in the running time, our $\sqrt{\Delta}\cdot \polylog\Delta+O(\log^* n)$-time \CONGEST algorithm for the $(\Delta+1)$-coloring problem also becomes computationally efficient. 

\end{document}